\documentclass[11pt]{article}
\usepackage{url,color}
\usepackage{graphicx,cite}
\usepackage{bbm,epsfig,epstopdf}
\usepackage{enumerate,lineno}
\usepackage{mathrsfs,amsfonts,amsmath,amsthm,amssymb}
\usepackage[normalem]{ulem}
\usepackage{todonotes}
\usepackage{CJK}
\usepackage{array}
\usepackage{xcolor}
\theoremstyle{plain}
\newtheorem{theo}{Theorem}[section]
\newtheorem{cor}[theo]{Corollary}
\newtheorem{rem}[theo]{Remark}
\newtheorem{defi}[theo]{Definition}
\newtheorem{lemma}[theo]{Lemma}

\newtheorem{prop}[theo]{Proposition}
\newtheorem{ex}[theo]{Example}

\newtheorem{op}[theo]{Open Problem}

\usepackage{algpseudocode}
\usepackage[Algorithmus,section]{algorithm}
\algnewcommand{\IfOneRow}[1]{\State\algorithmicif\ #1,}
\algnewcommand{\EndifOneRow}{}

\algdef{SE}[SUBALG]{Indent}{EndIndent}{}{\algorithmicend\ }%
\algtext*{Indent}
\algtext*{EndIndent}

\usepackage{caption}
\makeatletter
\renewcommand{\ALG@name}{Algorithm}
\makeatother

\numberwithin{theo}{section}
\numberwithin{equation}{section}
\numberwithin{table}{section}
\numberwithin{figure}{section}

\allowdisplaybreaks

\DeclareFixedFont{\ttb}{T1}{txtt}{bx}{n}{12} 
\DeclareFixedFont{\ttm}{T1}{txtt}{m}{n}{12}  

\usepackage{color}
\definecolor{deepblue}{rgb}{0,0,0.5}
\definecolor{deepred}{rgb}{0.6,0,0}
\definecolor{deepgreen}{rgb}{0,0.5,0}

\usepackage{listings}

\newcommand\pythonstyle{\lstset{
language=Python,
basicstyle=\ttm,
morekeywords={self},              
keywordstyle=\ttb\color{deepblue},
emph={MyClass,__init__},          
emphstyle=\ttb\color{deepred},    
stringstyle=\color{deepgreen},
frame=tb,                         
showstringspaces=false
}}

\lstnewenvironment{python}[1][]
{
\pythonstyle
\lstset{#1}
}
{}

\newcommand\pythoninline[1]{{\pythonstyle\lstinline!#1!}}

\usepackage{float}
\usepackage{tocloft}
\usepackage{amsfonts,amsmath,amssymb}
\usepackage{multirow, verbatim}
\usepackage{shadow,enumerate}
\usepackage{makeidx} 
\usepackage{graphicx}
\usepackage{color}

\newcommand{\vF}[1]{\mathbb{F}_2^{#1}}

\newcommand{\cB}{\mathcal{F}}
\newcommand{\VB}{\mathcal{VF}}
\usepackage{mathrsfs}

\usepackage{multirow,bm}
\usepackage{hhline}
\usepackage{array}
\usepackage{makecell}
\usepackage{multirow,bm}
\usepackage{hhline}
\usepackage{array}
\usepackage{makecell}
\def\cC{{\mathcal C}}
\def\cD{{\mathcal D}}
\def\cM{{\mathcal{MM}}}
\newcommand{\B}{\mathcal{B}}
\newcommand{\F}{\mathbb{F}}

\usepackage{booktabs}
\usepackage[colorlinks=true,citecolor=blue,linkcolor=blue,urlcolor=blue,bookmarks,bookmarksopen,bookmarksdepth=3,backref=page]{hyperref}
\usepackage{bookmark}
\bookmarksetup{
	numbered, 
	open,
}

\renewcommand*{\backref}[1]{}
\renewcommand*{\backrefalt}[4]{%
	\ifcase #1 (Not cited.)%
	\or        (Cited on page~#2.)%
	\else      (Cited on pages~#2.)%
	\fi}

\title{Design and analysis of bent functions using {$\mathcal{M}$}-subspaces}
\author{\Large Enes Pasalic$^1$, Alexandr Polujan$^2$, Sadmir Kudin$^1$, Fengrong Zhang$^{3,4}$\vspace{0.4cm} \\
	\small $^1$ University of Primorska, FAMNIT \& IAM, Glagolja\v{s}ka 8, 6000 Koper, Slovenia\ \\ \small \{\tt enes.pasalic6@gmail.com,
	\tt  sadmir.kudin@iam.upr.si\}\vspace{0.4cm}\\
	\small$^2$ Otto-von-Guericke-Universit\"{a}t, Universit\"{a}tsplatz 2, 39106, Magdeburg, Germany\ \\ \small \tt alexandr.polujan@gmail.com \vspace{0.4cm}\\
	\small	$^3$ State Key Laboratory of Integrated Services Networks,\ \\ Xidian University, Xian 710071, P.R. China \ \\
	\small $^4$ Mine Digitization Engineering Research Center of Ministry of Education,\ \\ China University of Mining and Technology, Xuzhou, Jiangsu 221116, China \ \\ \small \tt zhfl203@163.com
}
\date{}
\begin{document}
	
\maketitle
\begin{abstract}
	  In this  article, we provide the first systematic analysis of bent functions $f$ on $\F_2^{n}$ in the Maiorana-McFarland class $\cM$ regarding the origin and cardinality of their \textit{$\mathcal{M}$-subspaces}, i.e., vector subspaces on which the second-order derivatives of $f$ vanish. By imposing restrictions on  permutations $\pi$ of $\F_2^{n/2}$, we specify the conditions, such that Maiorana-McFarland bent functions  $f(x,y)=x\cdot \pi(y)  + h(y)$ admit a unique $\mathcal{M}$-subspace of dimension $n/2$. On the other hand, we show that permutations $\pi$ with linear structures give rise to Maiorana-McFarland bent functions that do not have this property. In this way, we contribute to the classification of Maiorana-McFarland bent functions, since the number of $\mathcal{M}$-subspaces is invariant under equivalence. Additionally, we give  several  generic methods of specifying  permutations $\pi$ so that $f\in\cM$ admits a unique $\mathcal{M}$-subspace. Most notably, using the knowledge about $\mathcal{M}$-subspaces, we show that using the bent 4-concatenation of four suitably chosen  Maiorana-McFarland bent functions, one can in a generic manner generate bent functions on $\F_2^{n}$ outside the completed Maiorana-McFarland class $\cM^\#$ for any even $n\geq 8$. Remarkably, with our construction methods it is possible to obtain inequivalent bent functions on $\F_2^8$ not stemming from two primary classes, the partial spread class $\mathcal{PS}$ and $\cM$. In this way, we contribute to a better understanding of the origin of bent functions in eight variables, since only a small fraction, of which size is about  $2^{76}$,  stems from $\mathcal{PS}$ and $\cM$, whereas the total number of bent functions on $\F_2^8$ is  approximately $2^{106}$. \ \\[1mm]
	  \textbf{Keywords.} Bent function, Maiorana-McFarland class, Partial spread class, Equivalence, Linear structure, Permutation, Bent 4-concatenation.
	\end{abstract}
\section{Introduction}
Bent functions are famous combinatorial objects introduced by Rothaus~\cite{Rot} in the mid-1960s that give rise to various discrete structures.  
Two known primary classes of bent functions are the Maiorana-McFarland class $\cM$ and the Partial Spread class $\mathcal{PS}$, which were introduced in the 1970s in \cite{MM73} and \cite{Dillon}, respectively. On the other hand, the so-called secondary constructions (the reader is referred to \cite{Mesnager}) use the known bent functions for the purpose of constructing new ones. However, only a few sporadic works on bent functions analyze the class inclusion properly, being more focused on specifying explicit univariate/bivariate trace form or construction methods without being precise whether these functions might belong to $\cM$ class for instance. This eventually leads to a lack of understanding related to the classification and enumeration of bent functions.  For instance, bent functions on $\F_2^8$ that belong to the main two primary classes are only a small fraction (about the size of $2^{76}$) of all $\approx 2^{106}$ bent functions in eight variables~\cite{LangevinL11}.

A pioneering work to provide bent functions that provably do not belong to $\cM$ or to $\mathcal{PS}$, up to equivalence, is due to Carlet \cite{CC93} who introduced two new classes of bent functions, the so-called $\cC$ and $\cD$ classes. In a recent series of articles \cite{Cclass_DCC,OutsideMM,BFAExtended,BapicPasalic,Kudin2021,Kudin2022}, the authors specified explicit families of bent functions  outside the completed $\cM$ class that belong to $\cC$ and $\cD$. Nevertheless, apart from the class $\cD_0$ of Carlet, these functions are defined on the variable space $n \geq 10$. Thus, the origin of bent functions outside $\cM^\# \cup \mathcal{PS}^\#$ on $\F_2^8$ is still unclear.  Moreover, most of the known secondary methods for constructing bent functions commonly employ bent functions on a smaller variable space. For example, in a recent article \cite{Bent_Decomp2022}, the authors provided several methods of generating infinite families of bent functions outside $\cM^\#$ using the so-called 4-concatenation $f=f_1||f_2||f_3||f_4$ of bent functions  $f_1,f_2,f_3,f_4$ in $n$ variables introduced in \cite{Decom} and later restated in \cite{SHCF}. Due to the design approach, employing bent functions outside $\cM^\#$ on a smaller space,  these results are significant only for $n \geq 10$ and do not answer the existence of bent functions outside the known primary classes	when $n=8$.  Such an approach then makes it impossible to construct bent functions on $\F_2^8$ since all bent functions in less than 8 variables are in $\cM^\#$. 

Dillon in his thesis~\cite{Dillon} proved that a given bent function $f$ on $\F_2^n$ belongs to the $\cM^\#$ class if and only if $D_aD_b f=0$ for all $a,b\in V$, where $V$ is a vector space of $\F_2^n$ of dimension $n/2$ (see also Lemma~\ref{lem M-M second} for details); these vector spaces were called \textit{$\mathcal{M}$-subspaces} in~\cite{Polujan2020}. Despite being introduced decades ago, the algebraic properties of $\mathcal{M}$-subspaces attracted attention only recently in a few works, e.g, in~\cite{Kolomeec17,Polujan2020,PolujanPhD}.

The main aim of this article is to provide the first systematic investigation of $\mathcal{M}$-subspaces of Boolean bent functions, and using this knowledge, provide generic construction methods of Boolean bent functions in $n$ variables outside the $\cM^\#$ class for all even $n\ge 8$. Notably, we give a characterization of bent functions on $\F_2^n$ in $\cM$ class, that have a unique $\mathcal{M}$-subspace $V=\F_2^{n/2} \times \{0_{n/2}\}$. We show that the property of a Maiorana-McFarland bent function $f(x,y)=x \cdot \pi(y) + h(y)$ to have a {\em unique $\mathcal{M}$-subspace} is, in many cases, completely determined by choice of permutation $\pi$.  In the other direction, if a permutation $\pi$ admits linear structures (implying that its  components also do) then $f \in \cM$ has at least two $\mathcal{M}$-subspaces. This characterization not only contributes to the classification of Maiorana-McFarland bent functions but also partially explains why the condition that the components of $\pi$ do not admit linear structures has  been efficiently used in, e.g., \cite{OutsideMM,Kudin2021,Cclass_DCC,BFAExtended,BapicPasalic} to specify functions in $\cC$ and $\cD$ that are outside $\cM^\#$.  More precisely, a modification of a bent function $f \in \cM$ is easier performed if only one vanishing subspace needs to be deprived of this property through the addition of  an indicator function. Using the obtained knowledge about $\mathcal{M}$-subspaces of Maiorana-McFarland bent functions, we provide several design methods of specifying bent functions $f_1,f_2,f_3,f_4$ on $\F_2^n$ such that the concatenation $f=f_1||f_2||f_3||f_4$ is bent on $\F_2^{n+2}$ and outside $\cM^\#$ for all $n\ge 6$. 
Additionally, we indicate that obtained with our approach bent functions on $\F_2^8$ are outside the $\mathcal{PS}^\#$ class as well, thus we contribute to the better understanding of the origin of all bent functions in $n=8$ variables.

The rest of the paper is organized in the following way. In Subsection~\ref{sec:pre} we recall basic definitions related to Boolean functions, and in Subsection~\ref{sub: bent-4 cat and its prop} we summarize the necessary algebraic properties of bent 4-concatenation. In Section~\ref{sec: 2}, we investigate, which classes of permutations $\pi$ on $\F_2^m$ are suitable for the construction of Maiorana-McFarland bent functions of the form $(x,y)\in\F_2^m\times\F_2^m\mapsto x\cdot \pi(y)$ with several $\mathcal{M}$-subspaces. Particularly, in Subsections~\ref{sub: 2.1} and~\ref{sub: 2.2}, we show that permutations with linear structures as well as quadratic permutations that admit many  $\mathcal{M}$-subspaces, respectively, lead to Maiorana-McFarland bent functions with several $\mathcal{M}$-subspaces. In Section~\ref{sec: 3}, we study the opposite question, namely,  we investigate, which classes of permutations $\pi$ on $\F_2^m$ are suitable for the construction of Maiorana-McFarland bent functions of the form $(x,y)\in\F_2^m\times\F_2^m\mapsto x\cdot \pi(y)+h(y)$ with the unique canonical $\mathcal{M}$-subspace. In Subsection~\ref{sub: 3.1}, we introduce permutations with the~\eqref{eq: P1} property as those permutations  $\pi$ on $\F_2^m$ for which $D_vD_w\pi\neq0_m$  for all linearly independent $v,w\in\F_2^m$. Remarkably, we show that permutations $\pi$ with this property guarantee that Maiorana-McFarland bent functions of the form $(x,y)\in\F_2^m\times\F_2^m\mapsto x\cdot \pi(y) + h(y)$ have  the unique canonical $\mathcal{M}$-subspace independently on the choice of a Boolean function $h$ on $\F_2^m$; the latter provides a variety of different Maiorana-McFarland bent functions with the unique $\mathcal{M}$-subspace even from a single permutation $\pi$ with this property. In Subsection~\ref{sub: 3.2}, we consider permutations  $\pi$ on $\F_2^m$ for which $D_uD_w \pi= 0_m$, for any $u,v\in S$,  where $\dim(S)\geq 1$. Remarkably, we completely characterize such permutations $\pi$ on $\F_2^m$ giving rise to bent functions $(x,y)\in\F_2^m\times\F_2^m\mapsto x\cdot \pi(y)$ with the unique canonical $\mathcal{M}$-subspace and refer to them as permutations with~\eqref{eq: P2} property in the sequel. In Section~\ref{sec: 4} we give several explicit constructions of permutations with~\eqref{eq: P1} and~\eqref{eq: P2} properties. In Section~\ref{sec: 5}, we provide several generic construction methods of bent functions outside the $\cM^\#$ class using the bent 4-concatenation. First, in Subsection~\ref{sub: 5.1}, we completely describe possible $\mathcal{M}$-subspaces of the bent 4-concatenation of four Maiorana-McFarland bent functions. Additionally, we explain how to check the membership in the $\mathcal{PS}^\#$ computationally. Consequently, we consider two different scenarios of the concatenation of Maiorana-McFarland bent functions which both lead to bent functions outside $\cM^\#$. In Subsection~\ref{sub: 5.2}, we show that if Maiorana-McFarland bent functions do not share a common $\mathcal{M}$-subspace, then their concatenation is outside $\cM^\#$. In subsection~\ref{sec:sharingMsubspace}, we show that even if Maiorana-McFarland bent functions share a common $\mathcal{M}$-subspace, then under certain technical conditions it is still possible that their concatenation is outside $\cM^\#$. Moreover, we indicate that with our approaches it is possible to construct inequivalent bent functions on $\F_2^8$ outside $\cM^\#\cup\mathcal{PS}^\#$.
In Section~\ref{sec: concl}, we conclude the paper and give a list of open problems.

\subsection{Preliminaries}\label{sec:pre}

The vector space $\mathbb{F}_2^n$ is the space of all $n$-tuples $x=(x_1,\ldots,x_n)$, where $x_i \in \mathbb{F}_2$.
For $x=(x_1,\ldots,x_n)$ and $y=(y_1,\ldots,y_n)$ in $\mathbb{F}^n_2$, the usual scalar (or dot) product over $\mathbb{F}_2$ is defined as $x\cdot y=x_1 y_1 + \cdots +  x_n y_n.$ The Hamming weight of $x=(x_1,\ldots,x_n)\in \mathbb{F}^n_2$ is denoted and computed as $wt(x)=\sum^n_{i=1} x_i.$ Throughout the paper, we denote by $0_n=(0,0,\ldots,0)\in \mathbb{F}^n_2$ the all-zero vector with $n$ coordinates, and by $\mathbbm{e}_{k}\in\F_2^n$ the $k$-th canonical basis vector. In certain cases, we endow $\F_2^n$ with the structure of the finite field $\left(\F_{2^{n}},\cdot \right)$. An element $\alpha \in \mathbb{F}_{2^n}$ is said to be a \emph{primitive element}, if it is a generator of the multiplicative group $\mathbb{F}_{2^n}^*$. The \emph{absolute trace} $Tr\colon \mathbb{F}_{2^n} \rightarrow \mathbb{F}_{2}$ is given by $Tr(x) =\sum_{i=0}^{n-1} x^{2^{i}}$.

The set of all Boolean functions in $n$ variables, which is the set of mappings from $\mathbb{F}_2^n$ to $\mathbb{F}_2$, is denoted by $\mathcal{B}_n$. 
It is well-known that any Boolean function $f\in\mathcal{B}_n$ can be uniquely represented by the \textit{algebraic normal form (ANF)}, which is given by $f(x_1,\ldots,x_n)=\sum_{u\in \mathbb{F}^n_2}{\lambda_u}{(\prod_{i=1}^n{x_i}^{u_i})}$, where $x_i, \lambda_u \in \mathbb{F}_2$ and $u=(u_1, \ldots,u_n)\in \mathbb{F}^n_2$.
The \textit{algebraic degree} of $f$, denoted by $\deg(f)$, is the maximum Hamming weight of $u \in \F_2^n$ for which $\lambda_u \neq 0$ in its ANF.

The \textit{first order-derivative} of a function $f\in\mathcal{B}_n$ in the direction $a \in \F_2^n$ is the mapping $D_{a}f(x)=f(x+a) +  f(x)$. Derivatives of higher orders are defined recursively, i.e., the \emph{$k$-th order derivative} of a function $f\in\mathcal{B}_n$ is defined by $D_Vf(x)=D_{a_k}D_{a_{k-1}}\ldots D_{a_1}f(x)=D_{a_k}(D_{a_{k-1}}\ldots D_{a_1}f)(x)$, where $V=\langle a_1,\ldots,a_k \rangle$ is a vector subspace of $\F_2^n$ spanned by elements $a_1,\ldots,a_k\in\F_2^n$. An element $a\in\F_2^n$ is called a \textit{linear structure} of $f\in\mathcal{B}_n$, if $f(x+a)+f(x)=const$ for all $x\in\F_2^n$. We say that  $f\in\mathcal{B}_n$ \textit{has no linear structures}, if $0_n$ is the only linear structure of $f$. 

The \emph{Walsh-Hadamard transform} (WHT) of $f\in\mathcal{B}_n$, and its inverse WHT, at any point $a\in\mathbb{F}^n_2$ are defined, respectively, by
\begin{equation*}
	W_{f}(a)=\sum_{x\in \mathbb{F}_2^n}(-1)^{f(x) +  a\cdot x} \quad\mbox{and}\quad
	(-1)^{f(x)}=2^{-n}\sum_{a\in \mathbb{F}_2^n}W_f(a)(-1)^{a\cdot x}.
\end{equation*}

For even $n$, a function $f\in\mathcal{B}_n$ is called {\em bent} if $W_f(u)=\pm2^{\frac{n}{2}}$ for all $u\in\F_2^n$. For a bent function $f\in\mathcal{B}_n$, a Boolean function $f^*\in \mathcal{B}_n$ defined by $W_f(u)=2^{\frac{n}{2}}(-1)^{f^*(u)}$ for all $u\in\F_2^n$ is a bent function, called the {\it dual} of $f$. Two Boolean functions $f,f'\in\mathcal{B}_n$ are called \textit{extended-affine equivalent}, if there exists an affine permutation $A$ of $\F_2^n$ and affine function $l\in\mathcal{B}_n$, such that $f\circ A + l= f'$. It is well known, that extended-affine equivalence preserves the bent property. In the sequel, while saying two Boolean functions are (in)equivalent, we always mean extended-affine equivalence, since this is the only type of equivalence we deal with in this article. 

The \textit{Maiorana-McFarland class} $\cM$ is the set of $n$-variable ($n=2m$) Boolean bent
functions of the form
\[
f(x,y)=x \cdot \pi(y)+ h(y), \mbox{ for all } x, y\in\F_2^m,
\]
where $\pi$ is a permutation on $\F_2^m$, and $h$ is an arbitrary Boolean function on
$\F_2^m$.
\begin{defi} 
	A class of bent functions $\mathit{B}_n \subset \mathcal{B}_n$ is \emph{complete} if it is globally invariant under extended-affine equivalence. The \emph{completed class}, denoted by $\cM^\#$ in the case of the Maiorana-McFarland class  $\cM$,  is the smallest possible complete class that  contains the class under consideration.
\end{defi}
With the following criterion of Dillon, one can show that a given Boolean bent function $f\in\mathcal{B}_n$ is (not) a member of the completed Maiorana-McFarland class.
\begin{lemma} \cite[p. 102]{Dillon}\label{lem M-M second}
	Let $n=2m$. A Boolean bent function $f\in\mathcal{B}_n$ belongs to $\cM^{\#}$ if and only if
	there exists an $m$-dimensional linear subspace $V$ of $\F_2^n$ such that the second-order derivatives
	$ D_{a}D_{b}f(x)=f(x) +  f(x +  a) +  f(x +  b) +  f(x +  a +  b)$
	vanish for any $ a,  b \in V$.
\end{lemma} 
Following the terminology in~\cite{Polujan2020}, we introduce the $\mathcal{M}$-subspaces of Boolean (not necessarily bent) functions in the following way.
\begin{defi}
	Let $f\in\mathcal{B}_n$ be a Boolean function. We call a vector subspace $V$ of $\F_2^n$  an $\mathcal{M}$-subspace of $f$, if for any $ a,  b \in V$ we have that 	$ D_{a}D_{b}f=0$. We denote by $\mathcal{MS}_r(f)$ the collection of all $r$-dimensional $\mathcal{M}$-subspaces of the function $f$.
\end{defi} 
It is well known~\cite{Carlet2021}, that for a bent function $f\in\mathcal{B}_n$ the maximum dimension of an $\mathcal{M}$-subspace is $n/2$; bent functions achieving this bound with equality are exactly the bent functions in $\cM^\#$ by Lemma~\ref{lem M-M second}. For every Maiorana-McFarland bent function $f(x,y)=x\cdot\pi(y)+h(y)$ on $\F_2^m\times\F_2^m$, the vector space $\F_2^{m} \times \{0_m \}$ is an $\mathcal{M}$-subspace of $f$, as observed by Dillon~\cite{Dillon}. However, in general, this vector space  $\F_2^{m} \times \{0_m \}$, which we refer to as \textit{the canonical $\mathcal{M}$-subspace},  is not necessarily unique. For instance, for a Maiorana-McFarland bent function $f$ on  $\F_2^{m} \times \F_2^{m}$, the number of its $\mathcal{M}$-subspaces is at most $\prod_{i=1}^{m} \left(2^i+1\right)$. Moreover, the equality is attained if and only if $f\in\mathcal{B}_{2m}$ is quadratic, as it was deduced in~\cite{PolujanPhD} from~\cite[Theorem 2]{Kolomeec17}. Finally, we note that in~\cite[Proposition 4.4]{Polujan2020} it was shown that the number of $\mathcal{M}$-subspaces of a Boolean function $f\in\mathcal{B}_n$ is invariant under equivalence; consequently, two bent functions with a different number of $\mathcal{M}$-subspaces are inequivalent. One can determine all $\mathcal{M}$-subspaces of a Boolean function $f\in\mathcal{B}_n$ as described in~\cite[Algorithm 1]{Polujan2020}.

We note that for vectorial functions, i.e., the mappings $F\colon\F_2^n\to\F_2^m$, one can essentially extend the definitions related to differential properties (e.g., derivatives, linear structures and $\mathcal{M}$-subspaces) by simply replacing $f\in\mathcal{B}_n$ by  $F\colon\F_2^n\to\F_2^m$ in the corresponding definitions.  For $b\in\F_2^m$, the \textit{component function} $F_b\in\mathcal{B}_n$ of $F\colon\F_2^n\to\F_2^m$ is defined by $F_b(x)=b\cdot F(x)$ for all $x\in\F_2^n$. Finally, every vectorial function $F\colon\F_2^n\to\F_2^m$ can be uniquely represented in the form $F(x)=(f_1(x),\ldots,f_m(x))^T$, where Boolean functions $f_i\in\mathcal{B}_n$ are called the \textit{coordinate functions} of $F$; thus the algebraic normal form and the algebraic degree of $F$ are defined coordinate-wise.

\subsection{Bent 4-concatenation and its algebraic properties}\label{sub: bent-4 cat and its prop}
In the following, we will be mainly interested in the design of bent functions $f\in\mathcal{B}_{n+2}$ from four bent functions $f_1,f_2,f_3,f_4\in\mathcal{B}_n$ using the \textit{bent 4-concatenation} $f=f_1||f_2||f_3||f_4$, of which ANF is given by
\begin{equation}\label{eq:ANF_4conc}
f(x,y_1,y_2)=f_1(x) +  y_1(f_1 +  f_3)(x) +  y_2(f_1 +  f_2)(x) +  y_1y_2(f_1 +  f_2 +  f_3 +  f_4)(x).
\end{equation}
From this expression, it is not difficult to see that  $f_1(x)=f(x,0,0), f_2(x)=f(x,0,1),f_3(x)=f(x,1,0)$ and $f_4(x)=f(x,1,1)$. Note that if $f_i \in \B_n$ are all bent, then the necessary and sufficient condition that $f=f_1||f_2||f_3||f_4 \in \mathcal{B}_{n+2}$ is bent as well, is that the \textit{dual bent condition} is satisfied~\cite{SHCF}, i.e., $f^*_1  +  f^*_2  +  f^*_3  +  f^*_4=1$.

For the further analysis of the bent 4-concatenation $f=f_1||f_2||f_3||f_4$ in terms of the second-order derivatives, we derive the expression for $D_aD_b f(x,y_1,y_2)$ where  $a=(a',a_1,a_2)$ and $b=(b',b_1,b_2)$ and $a',b' \in \F_2^n$ and $a_i,b_i \in \F_2$ as follows:
\begin{equation}\label{eq:2ndderiv_conc correct}
	\begin{split}
		D_aD_bf(x,y_1,y_2)&=D_{a'}D_{b'}f_1(x)+ y_1D_{a'}D_{b'}f_{13}(x)+ y_2 D_{a'}D_{b'}f_{12}(x)+  y_1y_2D_{a'}D_{b'}f_{1234}(x) \\
		& + a_1D_{b'}f_{13}(x +  a')+  b_1 D_{a'}f_{13}(x +  b') +  a_2D_{b'}f_{12}(x +  a') +  b_2D_{a'}f_{12}(x +  b')  \\
		& +  (a_1y_2 +  a_2y_1 +  a_1a_2)D_{b'}f_{1234}(x +  a') +  (b_1y_2 +  b_2y_1 +  b_1b_2)D_{a'}f_{1234}(x +  b') \\
		 &  +  (a_1b_2 +  b_1a_2)f_{1234}  (x +  a' +  b'), 
	\end{split}
\end{equation}
where the Boolean function $f_{i_1\ldots i_k}\in\mathcal{B}_n$ is defined by $f_{i_1\ldots i_k}:=f_{i_1} +  \cdots  +  f_{i_k}$.	

In this context, the main design goal is to specify suitable $f_i\in\mathcal{B}_n$ so that $f\in\mathcal{B}_{n+2}$ is a bent function, and to ensure that $f$ does not satisfy the $\cM^\#$ class membership criterion of Dillon.
\section{Bent functions with more than one $\mathcal{M}$-subspace}\label{sec: 2}

In this section, we derive sufficient conditions that $f(x,y)=x\cdot \pi(y)  + h(y)$ admits more than one $\mathcal{M}$-subspace. This feature is disadvantageous from the perspective of constructing bent functions $f=f_1||f_2||f_3||f_4\in\mathcal{B}_{2m+2}$ outside $\cM^\#$ from Maiorana-McFarland bent functions $f_i\in\mathcal{B}_{2m}$, since in this case, it is more difficult to ensure that the second-order derivatives of $f$ do not vanish on any $(m+1)$-dimensional subspace of $\F_2^{2m+2}$. Essentially, this property is closely related to the choice of a permutation $\pi$ on $\F_2^m$ which is then characterized by the presence of non-zero linear structures or being quadratic.
\subsection{Permutations with linear structures}\label{sub: 2.1}
First, we show that permutations with linear structures give rise to Maiorana-McFarland bent functions with more than one $\mathcal{M}$-subspace.
\begin{prop} \label{prop:uniqueMnecessary} Let $\pi$ be a permutation of $\F_2^m$ with a non-zero linear structure $s \in \F_2^m$, i.e., $$D_s\pi(x)= \pi(x) +  \pi(x+s)=v \in \F_2^m$$ holds for all $x\in\F_2^m$, and let $h: \F_2^m \to \F_2$ be an arbitrary Boolean function.
Then, the function $g\colon\F_2^{m}\times\F_2^{m} \to \F_2$ defined by $$g(x,y)=x \cdot \pi(y) +h(y),\quad\mbox{for all } x,y \in \F_2^{m},$$ 
has at least two $\mathcal{M}$-subspaces.
\end{prop}
\begin{proof} Clearly, the canonical $\mathcal{M}$-subspace $\F_2^{m} \times \{0_m \}$ is the first one. We will now construct another one. Let $v=D_s\pi \in \F_2^{m}$ and $W=\langle v \rangle ^{\perp} \subset \F_2^{m}$. Set $V=\langle W \times \{0_m \}, (0_m,s) \rangle$. 
	For two different non-zero vectors $a=(a_1,a_2)$ and $b=(b_1, b_2)$ in $V$ we compute 
	\begin{equation*}
		D_{a}D_{b}g(x)=x\cdot\left( D_{a_2}D_{b_2}\pi(y)\right)
		+  a_1\cdot D_{b_2}\pi(y+ a_2) 
		+   b_1 \cdot D_{a_2}\pi(y+ b_2)+D_{a_2}D_{b_2}h(y). \nonumber
	\end{equation*}
	If $a_2=b_2=0_m$, i.e, if $a$ and $b$ are in $W \times \{0_m \}$, we deduce that $D_{(a_1, a_2)}D_{(b_1, b_2)}g(x)=0$.
	If $b=(0_m,s)$ and $a \in W \times \{0_m \}$, then $a_2=0_m$, and we have 
	$$D_{(a_1, a_2)}D_{(b_1, b_2)}g(x)= a_1 \cdot D_{s}\pi(y)=a_1 \cdot v=0,$$
	since $a_1 \in W= \langle v \rangle ^{\perp}$.
	From this, we conclude that the second-order derivatives of $g$ vanish on $V$ as well.
\end{proof}

However, the condition that permutation $\pi$ of $\F_2^m$ has no linear structures does not imply that the only vanishing $\mathcal{M}$-subspace is $\F_2^m \times \{0_m\}$, as the following example shows.

\begin{ex} \label{ex:twoMsubspaces}
	Let $m=5$ and $\pi$ be a permutation of $\F_2^m$ defined by its algebraic normal form in the following way:
	\begin{equation}\label{eq: pi ugly}
		\pi(y)=\begin{bmatrix} y_1\\ y_2\\ y_3 + y_1 y_3 + y_1 y_5\\ y_1 y_3 + y_2 y_3 + y_4\\ 
			y_1 y_3 + y_2 y_4 + y_5 + y_1 y_5 \end{bmatrix}.
	\end{equation}
	It is not difficult to check, that the only linear structure of $\pi$ is $s=0$. However, the function $g(x,y)=x \cdot \pi(y)$ has exactly two $\mathcal{M}$-subspaces: the canonical $\mathcal{M}$-subspace $\F_2^{m} \times \{0_m \}$ as well as $V$, which is given by:
	\begin{equation*}
		V=
		\left\langle\scalebox{0.7}{$\begin{array}{cccccccccc}
				1 & 0 & 0 & 0 & 0 & 0 & 0 & 0 & 0 & 0 \\
				0 & 1 & 0 & 0 & 0 & 0 & 0 & 0 & 0 & 0 \\
				0 & 0 & 0 & 0 & 0 & 0 & 0 & 1 & 0 & 0 \\
				0 & 0 & 0 & 0 & 0 & 0 & 0 & 0 & 1 & 0 \\
				0 & 0 & 0 & 0 & 0 & 0 & 0 & 0 & 0 & 1 \\
			\end{array}$}
		\right\rangle.
	\end{equation*}
	Note that for the permutation $\pi$ defined in~\eqref{eq: pi ugly}, there exist a lot of Boolean functions $h$ on $\F_2^5$ such that by adding Boolean function $h(y)$ on $\F_2^5$ to $g(x,y)=x \cdot \pi(y)$, one gets a bent function $f(x,y)=x \cdot \pi(y) + h(y)$ having the unique canonical $\mathcal{M}$-subspace. A concrete example of such a function is $h(y_1,\ldots,y_5)=y_3 y_4 y_5$.
\end{ex}

\subsection{Quadratic permutations inducing  more than one  $\mathcal{M}$-subspace  for bent functions in $\cM$}\label{sub: 2.2}
In this subsection, we provide instances of quadratic permutations for which the function defined by $f(x,y)=x \cdot \pi  (y)$ has more than one $\mathcal{M}$-subspace. We will use the following two results from \cite{Kudin2022}.

\begin{lemma} \cite{Kudin2022} \label{lem:algdegvec}
	Let $G: \F_2^m \to \F_2^t$ be a vectorial Boolean function. If there exists an $(m-k)$-dimensional subspace $H$ of $\F_2^m$ such that $D_aD_b G=0_t$ for all $a,b \in H$, then the algebraic degree of $G$ is at most $k+1$.
\end{lemma}

\begin{lemma} \cite{Kudin2022} \label{lem:afcomponent} Let $\pi: \F_2^m \to \F_2^m$ be a permutation such that there is a linear hyperplane $V$ of $\F_2^m$, on which $\pi$ is affine. Let $l(x)$ be the linear Boolean function that defines $V$, that is,  $l(x)=0$ if and only if $x \in V$. Then, $l(x)$ or $l(x)+1$ is a component function of $\pi$.
\end{lemma}

\begin{lemma}\label{lem:afcomponent2} Let $\pi$ be a permutation of $\F_2^m$, such that there exists an $(m-1)$-dimensional subspace $S \subset \F_2^m$ for which $D_aD_b \pi=0_m$, for all $a,b \in S$. Let $s: \F_2^m \to \F_2$ be the linear Boolean function that defines $S$, that is,  $s(y)=0$ if and only if $y \in S$. Then, $\pi$ is at most quadratic and $s(y)$ or $s(y)+1$ is a component function of $\pi$.
\end{lemma}
\begin{proof}
	The fact that $\pi$ is at most quadratic follows directly from Lemma \ref{lem:algdegvec}. Let $a,b$ be two arbitrary vectors from $S$. Since $D_aD_b\pi(y)=0_m$ for all $y\in \F_2^m$, setting $y=0_m$ we get:
	$$\pi(a+b)+\pi(a)+\pi(b)+\pi(0_m)=0_m. $$
	Since $a,b \in S$ were arbitrary, we deduce that $\pi$ is affine on the linear hyperplane $S$, and from Lemma \ref{lem:afcomponent} it follows that $s(y)$ or $s(y)+1$ is a component function of $\pi$.
\end{proof}

\begin{prop} Let $\pi$ be a permutation of $\F_2^m$, such that there exists an $(m-1)$-dimensional subspace $S \subset \F_2^m$ for which $D_aD_b \pi=0_m$, for all $a,b \in S$.
	Let $f\colon \F_2^{m}\times \F_2^{m} \to \F_2$ be the function defined by:
	$$f(x,y)=x \cdot \pi(y).$$
	Then, $f$ has at least two $\mathcal{M}$-subspaces.
\end{prop}
\begin{proof} It is obvious that $\F_2^m \times \lbrace 0_m \rbrace$ is one $\mathcal{M}$-subspace for $f$.
	Let $s: \F_2^m \to \F_2$ be the linear Boolean function that defines $S$, that is,  $s(y)=0$ if and only if $y \in S$. From Lemma \ref{lem:afcomponent2} we deduce that $s(y)$ or $s(y)+1$ is a component function of $\pi$. Let $c \in \F_2^m$ be such that $c \cdot \pi$ is equal to $s$ or $s+1$. Let $S'$ denote the subspace $S'= \lbrace 0_m \rbrace \times S$, and let $V$ be the $m$-dimensional subspace of $\F_2^{2m}$ defined by $V= \langle (c,0_m),S' \rangle$. We will show that $V$ is also an $\mathcal{M}$-subspace for $f$. If $v=(v_1,v_2)$ and $w=(w_1,w_2)$ are from $V$ such that $v_1=w_1=0_m$, that is $v,w \in S'$, then $v_2,w_2$ are in $S$, and $$D_vD_wf(x,y)=x \cdot D_{v_2}D_{w_2}\pi(y)=0.$$
	Assume now that $v=(c,0_m)$ and $w \in S'$. Then
	$$D_vD_wf(x,y)= D_w(c\cdot \pi(y))= s(y+w_2)+s(y).$$
	Since $w_2$ is in $S$, then $y+w_2$ is in $S$ if and only if $y$ is in $S$, hence $s(y+w_2)=s(y)$, and consequently $$D_vD_wf(x,y)= s(y+w_2)+s(y)=0.$$
	We conclude that $D_vD_wf=0$ for all $v,w \in V$, and hence that $V$ is also an $\mathcal{M}$-subspace for~$f$.
\end{proof}

\section{Bent functions in $\cM$ with the unique canonical $\mathcal{M}$-subspace}\label{sec: 3}
In this section, we characterize more precisely permutations that give rise to the unique  canonical $\mathcal{M}$-subspace  for $f(x,y)=x\cdot\pi(y)+h(y)$.
This is achieved through two useful properties called $(P_1)$ and $(P_2)$ which classify  permutations with respect to vanishing subspaces of its second-order derivatives $D_aD_b \pi$.  In  Section \ref{sec: 4}, we will  provide some generic methods of specifying permutations satisfying these properties, including a  generic class of APN permutations that necessarily satisfy the $(P_1)$ property.
\subsection{Bent functions from permutations having $(P_1)$ property}\label{sub: 3.1}
In the following statement, we provide a sufficient condition on permutations $\pi$ of $\F_2^m$, such that $f(x,y)=x\cdot\pi(y)+h(y)$ has the unique $\mathcal{M}$-subspace $\F_2^m\times\{0_m\}$ independently on the choice of a function $h$ on $\F_2^m$.
\begin{theo}\label{theo: unique P1} Let $\pi$ be a permutation of $\F_2^m$ which has the following property:
	\begin{equation}\label{eq: P1} \tag{$P_1$}
		D_vD_w\pi\neq0_m \mbox{ for all linearly independent } v,w\in\F_2^m.
	\end{equation}
	Define $f\colon\F_2^{m}\times\F_2^{m} \to \F_2$ by $f(x,y)=x \cdot \pi(y) + h(y)$, for all $x,y \in \F_2^{m}$, where $h\colon\F_2^m \to \F_2$ is an arbitrary Boolean function. Then, the following hold:
	\begin{itemize}
		\item[1.] Permutation $\pi$ has no linear structures.
		\item[2.] The vector space $V=\F_2^m \times \{0_m \}$ is the only $\mathcal{M}$-subspace of $f$.
	\end{itemize}
\end{theo}
\begin{proof} 
	\emph{1.} Assume that $\pi$ has a non-zero linear structure $a\in\F_2^m$, i.e., for all $x\in\F_2^m$ holds $D_{a}\pi(x)=v$ for some $v\in\F_2^m$. Then, taking $b\in\F_2^m\setminus\{0_m,a\}$, we get that $D_{a}D_{b}\pi=0_m$, which contradicts the property~\eqref{eq: P1}. \ \\ 
	\noindent\emph{2.} Let $V$ be an $m$-dimensional subspace of $\F_2^{2m}$ such that $D_aD_bf=0$ for all $a,b \in V$. Define the linear mapping $L:V \to \F_2^m$ by $L(x,y)=y$, for all $(x,y) \in V$.
	
	In general, the second-order derivative of  $f$ is given by,
	\begin{equation} \label{eq: secder2}
		D_{(a_1, a_2)}D_{(b_1, b_2)}f(x,y)
		=x\cdot\left( D_{a_2}D_{b_2}\pi(y)\right)+  a_1\cdot D_{b_2}\pi(y+ a_2)
		+ b_1 \cdot D_{a_2}\pi(y+ b_2) + D_{a_2}D_{b_2}h(y) . 
	\end{equation}
	If $a_2,b_2 \in \F_2^m \setminus \{0_m \}$ and $a_2 \neq b_2$, then $D_{a_2}D_{b_2}\pi(y) \neq 0_m$, so $D_{(a_1, a_2)}D_{(b_1, b_2)}f \neq 0$, because $x\cdot\left( D_{a_2}D_{b_2}\pi(y)\right) \neq 0$. Since for all $a,b \in V$ we have $D_aD_bf=0$, we deduce that, for all $a=(a_1,a_2),b=(b_1,b_2)$ in $V$, either  $L(a)=a_2=0_m$, or $L(b)=b_2=0_m$, or $L(a)=a_2=b_2=L(b)$.
	This means that $\dim(Im(L)) \leq 1$. From the rank-nullity theorem, we get that $\dim(Ker(L))\geq m-1$. If $\dim(Ker(L))=m$, then $V=\F_2^m \times \{0_m \}$. 
	
	Assume now that $\dim(Ker(L))=m-1$, and let $b=(b_1,b_2) \in V$ be the vector such that $b_2\neq0_m$. For all $a=(a_1,a_2) \in Ker(L)$ we have $a_2=0$, and hence 
	\begin{equation} \label{a2iszero}
		D_{(a_1, a_2)}D_{(b_1, b_2)}f(x,y) = a_1 \cdot D_{b_2} \pi(y)=0, \text{ for all } y \in \F_2^m.
	\end{equation}
	Denote by $S_b$ the subspace of $\F_2^m$ generated by $\{D_{b_2} \pi(y) \colon y \in \F_2^m \}$. Not that, since $\pi$ is a permutation, and $b_2\neq 0_m$ the vector $D_{b_2} \pi(y)=\pi(y)+\pi(y + b_2)$ is never equal to $0_m$, this means that if $\dim(S_b)=1$, then $D_{b_2} \pi(y)$ is constant (i.e.,\ $b_2$ is a linear structure for $\pi$), and consequently, for any nonzero $c \in \F_2^m \setminus \{0_m,b_2 \}$, we have $D_cD_{b_2}\pi=0_m$, which is in contradiction with the assumption $D_vD_w\pi \neq 0_m$, for all nonzero different $v,w \in \F_2^m$. This implies that $\dim(S_b) \geq 2$, and hence $\dim(S_b^{\perp})\leq m-2$. From the equation \eqref{a2iszero} we have that for every $a=(a_1,a_2) \in Ker(L)$, the vector $a_1$ is in $S_b^{\perp}$, hence $\{ a_1 \colon a=(a_1,a_2) \in Ker(L) \} \subseteq S_b^{\perp}$. However, $\dim( \{ a_1 \colon a=(a_1,a_2) \in Ker(L) \} )=\dim(Ker(L))=m-1$, and this is a contradiction, because $\dim(S_b^{\perp})\leq m-2$. This means that the case $\dim(Ker(L))=m-1$ is not possible, hence, the only $m$-dimensional subspace of $\F_2^{2m}$ such that $D_aD_bf=0$ for all $a,b \in V$, is $V=\F_2^m \times \{0_m \}$.
\end{proof}

Imposing an additional condition on the permutation $\pi$, it is possible to further refine the  structure of vanishing subspaces.

\begin{cor}\label{cor:dim2}  Let $\pi$ be a permutation of $\F_2^m$ with the property~\eqref{eq: P1} and such that $\gamma\cdot \pi$ has no nonzero linear structures for $\gamma\in \F_2^{m}\setminus \{0_m\}$. Let $f\colon\F_2^{m}\times\F_2^{m} \to \F_2$ be the function defined by $f(x,y)=x \cdot \pi(y) + h(y)$, for all $x,y \in \F_2^{m}$, where $h: \F_2^m \to \F_2$ is an arbitrary Boolean function. If $S$ is a subspace of $\F_2^{m}\times \F_2^{m}$ such that $\dim(S)>1$ and $D_aD_bf=0$, for all $a,b \in S$, then $S$ is a subspace of $\F_2^{m} \times \{0_m \}$.
\end{cor}
\begin{proof}
	Notice that since $\pi$ has the~\eqref{eq: P1} property, there exist no two distinct nonzero elements
	$u, v \in\F_2^m$ such that $D_uD_v\pi(y)=0_m$, for all $y \in \F_2^m$.
	 Consequently, $\pi(y)+\pi(y+u) + \pi(y+v) +  \pi(y+u+v)\neq 0_m$ for any distinct  nonzero $u,v \in\F_2^m$. Then,  denoting $a=(a_1,a_2)$, $  b=(b_1,b_2)\in\F_2^m \times \F_2^m$, we  have   
	\begin{equation*}
			D_{(a_1, a_2)}D_{(b_1, b_2)}f(x,y)
			=x\cdot\left( D_{a_2}D_{b_2}\pi(y)\right)
			+  a_1\cdot D_{b_2}\pi(y+ a_2) 
			+   b_1 \cdot D_{a_2}\pi(y+ b_2)
			+D_{a_2}D_{b_2}h(y). 
	\end{equation*} 
	The term $x\cdot\left( D_{a_2}D_{b_2}\pi(y)\right)$ cannot be cancelled unless $a_2=0_m$  or $b_2=0_m$, alternatively  $a_2=b_2 \neq 0_m$. 
	Assuming  that  $a_2=0_m$ and $b_2 \neq 0_m$ (the same  reasoning applies if $b_2=0_m$) leads to $D_{(a_1, a_2)}D_{(b_1, b_2)}f(x,y)=  a_1\cdot D_{b_2}\pi(y) $ which implies that $a_1=0_m$ and therefore $a=(a_1,a_2)=0_m,0_m)$, a contradiction. The case $a_2=b_2 \neq 0_m$, implying also that $a_1 \neq b_1$ since $\dim(S) >1$, gives $D_{(a_1, a_2)}D_{(b_1, b_2)}f(x,y)=(a_1+ b_1) \cdot D_{a_2}\pi(y+ a_2)$ which is nonzero (since $a_1+b_1 \neq 0_m$) and consequently $D_{(a_1, a_2)}D_{(b_1, b_2)}f(x,y)\neq 0$.
\end{proof}


The following result specifies  both the  necessary and sufficient condition for a permutation $\pi$ on $\F_2^m$, when the function $h(y)=\delta_0(y)=\prod_{i=1}^m(y_i+1)$ is used to define $f(x,y)= x \cdot \pi(y) + h(y)$, so that $f$  admits only the canonical vanishing  $\mathcal{M}$-subspace $\F_2^m \times \{0_m\}$.  
\begin{prop}
Let $\pi$ be a permutation of $\F_2^m$ with $\deg(\pi)<m-1$, and let $f\colon \F_2^{m}\times\F_2^{m} \to \F_2$ be the function defined by 
$$f(x,y)= x \cdot \pi(y) + \delta_0(y), \text{ for all } x,y \in \F_2^{m}.$$
Then $f$ has only one $\mathcal{M}$-subspace if and only if $\pi$ has no  nonzero linear structures.
\end{prop}
\begin{proof}
If $\pi$ has linear structures, then the fact that $f$ has at least two $\mathcal{M}$-subspaces follows from Proposition \ref{prop:uniqueMnecessary}.

Assume now that $\pi$ has no nonzero linear structures. Let $V$ be an $m$-dimensional subspace of $\F_2^{2m}$ such that $D_aD_bf=0$ for all $a,b \in V$. Define the linear mapping $L:V \to \F_2^m$ by $L(x,y)=y$, for all $(x,y) \in V$.
	In general, the second-order derivative of  $f$, for any $a_1,a_2,b_1,b_2 \in \F_2^m$, is given by
	\begin{equation}\label{secderdelta0}
		D_{(a_1, a_2)}D_{(b_1, b_2)}f(x,y)
		=x\cdot\left( D_{a_2}D_{b_2}\pi(y)\right)+  a_1\cdot D_{b_2}\pi(y+ a_2)
		+ b_1 \cdot D_{a_2}\pi(y+ b_2) + D_{a_2}D_{b_2}\delta_0(y) . 
	\end{equation}

Assume that $\dim(Im(L))\geq 2$. Let $(c_1,c_2), (d_1,d_2) \in V$ be such that $c_2$ and $d_2$ are two different nonzero elements in $\F_2^m$. Since the algebraic degree of $D_{c_2}D_{d_2}\delta_0(y)$ is $m-2$, and since $\deg(\pi)<m-1$, from \eqref{secderdelta0} we deduce that the algebraic degree of $D_{(c_1, c_2)}D_{(d_1, d_2)}f$ is $m-2$, and that is a contradiction, since $(c_1,c_2), (d_1,d_2) \in V$ and so $D_{(c_1, c_2)}D_{(d_1, d_2)}f=0$.

If $\dim(Im(L))=1$, then $\dim(Ker(L))=m-1$. Let $(a_1,a_2) \in V$ be such that $a_2 \neq 0_m$, and let $(b_1,0_m) \in V$ be an arbitrary element in $Ker(L)$. From \eqref{secderdelta0} we compute
	\begin{equation*}
		D_{(a_1, a_2)}D_{(b_1, 0_m)}f(x,y)
		= b_1 \cdot D_{a_2}\pi(y)=0, \text{ for all } x,y \in \F_2^m. 
	\end{equation*}
This means that the subspace $S_{a_2}$ generated by the set $ \lbrace D_{a_2}\pi(y) \colon y \in \F_2^m \rbrace$ is in the orthogonal complement of $b_1$, for every $b_1$ such that $(b_1,0_m) \in Ker(L)$. Since $\dim(Ker(L))=m-1$, we deduce that $\dim(S_{a_2})=1$. Also, $\pi$ is a permutation and $a_2 \neq 0_m$, so $D_{a_2}\pi(y) \neq 0_m$, for all $y \in  \F_2^m$, hence $ \lbrace D_{a_2}\pi(y) \colon y \in \F_2^m \rbrace = \lbrace v \rbrace$ for some nonzero $v \in \F_2^m$, and this means that $a_2$ is a nonzero linear structure of $\pi$. However, this is a contradiction, since the assumption is that $\pi$ has no nonzero linear structures.

We conclude that it has to be the case that $\dim(Im(L))=0$, and consequently that the only $\mathcal{M}$-subspace of $f$ is $V= \F_2^m \times \lbrace 0_m \rbrace$.
\end{proof}


\subsection{Bent functions from permutations having $(P_2)$ property}\label{sub: 3.2}
In the following statement, we show that even permutations on $\F_2^m$, for which second-order derivatives vanish on a certain $(m-k)$-dimensional subspace $S$ (where $2 \leq k  \leq m-1$), can  still be used for the construction of Maiorana-McFarland bent functions with a unique $\mathcal{M}$-subspace.
\begin{prop}\label{prop:suffcondunique}	
	Let $\pi$  be a nonlinear permutation over $\F_2^m$ and $f(x,y)= x \cdot \pi(y)$ a bent function in $\cM$. Denote by $S$ a vector subspace of $\F_2^m$ for which
	$D_aD_b \pi(y) = 0$, for any $a,b\in S$,  where $\dim(S) \geq 1$. If $\dim(S)=m-k$, then the necessary and sufficient condition for $f$ to have the unique canonical $\mathcal{M}$-subspace is that there do not exist  linearly independent 
		$u_1, \ldots, u_k \in \F_2^m$ for which $u_i \cdot D_a \pi(y)=0$ for any $a \in S$, and we necessarily have that $2 \leq k  \leq m-1$.
\end{prop}
\begin{proof}
	It is clear that if $\pi$ is linear/affine then $\dim(S)=m$ and the number of $\mathcal{M}$-subspaces is $\prod_{i=1}^{m} \left(2^i+1\right)$. Thus, we need to show that $\dim(S)$ cannot be $m-1$.  Assuming that   $\dim(S)=m-1$,  Lemma \ref{lem:algdegvec} and Lemma \ref{lem:afcomponent2} imply that $\pi$ is at most quadratic  and affine on this hyperplane determined by $S$. Furthermore,  there exists $u_1$ such that $u_1 \cdot D_a \pi(y)=0$. 
	Noticing that $D_{a_2}D_{b_2} \pi(y)=0$ for any $a_2,b_2 \in S$,  
	let  $S' =\{0_m\} \times S$ be a subspace of $\F_2^m \times \F_2^m$ of dimension $m-1$. Then,  for any $a=(a_1,a_2),b=(b_1,b_2) \in S'$  
	\begin{equation} \label{eq: secorder}
			D_{(a_1, a_2)}D_{(b_1, b_2)}f(x,y)
			=x\cdot\left( D_{a_2}D_{b_2}\pi(y)\right)
			+  a_1\cdot D_{b_2}\pi(y+ a_2) 
			+   b_1 \cdot D_{a_2}\pi(y+ b_2)=0,
	\end{equation} 
	since $a_1=b_1=0$. Then, adjoining $(u_1,0)$ to $S'$ so that $S=\langle (u_1,0), S' \rangle$, we would have that $\dim(S)=m$ and $D_{(u_1, 0)}D_{(b_1, b_2)}f(x,y)=0$ for any $(b_1, b_2) \in S'$ (where $b_1=0$). Consequently, $S$ is a  vanishing subspace for $f$ and different from $\F_2^m \times 0$. Thus, to  have the unique vanishing subspace we necessarily have that $\dim(S)\leq m-2$, that is $ k \geq 2$.
	
	
	 In general,  when $\dim(S)=m-k$ where $2 \leq k  \leq m-1$ a similar reasoning applies. Extending $S' =  \{0_m\} \times S$ to the full dimension $m$, by adjoining $(u_1,0_m), \ldots, (u_k,0_m)$ to $S'$, is impossible due to our assumption. This follows from the fact that taking, e.g., $(u_1,0)$ and $(b_1,b_2) \in S'$ (where $b_1=0$), the equation \eqref{eq: secorder} reduces to $u_1 \cdot D_{b_2}\pi(y)$, which is nonzero.  On the other hand, we can also extend $S'$ by adjoining elements in $(b_1,b_2) \in \F_2^m \times \F_2^m$  where $b_2 \in S$, which is necessary  for   ensuring that $x\cdot\left( D_{a_2}D_{b_2}\pi(y)\right)$ is cancelled if we consider $(a_1,a_2)$ and $(b_1,b_2)$, where $a_2  \neq  b_2 \in S$. However, adjoining $(b_1,b_2)$ to $S'$ 
	 implies that $(u_i, 0) \in \langle (b_1,b_2), S' \rangle $ and the same reasoning as above applies. 
\end{proof}
We state  this   property more  formally  in the following definition.
\begin{defi}
	Let $S$ be any subspace of dimension $m-k$, with $2 \leq k  \leq m-1$, such that $D_aD_b\pi(y)=0_m$ for all $a,b \in S$, where $\pi$ is a nonlinear permutation on $\F_2^m$. Then, $\pi$ satisfies the property ($P_2$) with respect to $S$ if:
		\begin{equation}\label{eq: P2} \tag{$P_2$}
	\dim(S)=m-k \textnormal{ with } 2 \leq k  \leq m-1; \not \exists u_1, \ldots, u_k \in \F_2^m: u_i \cdot D_a \pi(y)=0 \textnormal{ for all } a \in S.
	\end{equation}
	If $\pi$ satisfies this property  with respect to any $S$ of arbitrary   dimension $1 \leq \dim(S)  \leq m-2$, then we simply say that $\pi$ (fully) satisfies \eqref{eq: P2}.  
\end{defi}
\begin{rem}
	For instance, the permutation  $\pi$ on $\F_2^5$ from Example~\ref{ex:twoMsubspaces} does not satisfy the conditions in Proposition \ref{prop:suffcondunique}.  Here  $\dim(S)=m-2=3$ and 
	two vectors $u_1=((1,0,0,0,0),0_5)$ and $u_1=((0,1,0,0,0),0_5)$ can be adjoined to $S'=\{0_5\} \times S$ since they select linear functions $y_1$ and $y_2$ whose first order derivatives vanish for any choice of $a_2 \in S$.  
\end{rem}  	

\begin{rem}
	1. Note that the property \eqref{eq: P1} implies \eqref{eq: P2}, but not vice versa.

	\noindent 2. As shown in~\cite{BBS2017}, there exist 75 affine inequivalent quadratic permutations $\pi$ of $\F_2^5$. Among them, 34 permutations give rise to bent functions $(x,y)\mapsto x\cdot \pi(y)$ with the unique canonical $\mathcal{M}$-subspace. With respect to the properties \eqref{eq: P1}, \eqref{eq: P2}, they are distributed as follows:
	\begin{itemize}
		\item 2 permutations have the property \eqref{eq: P1}, note that these permutations are APN;
		\item 32 permutations have the property \eqref{eq: P2} (but not \eqref{eq: P1}).
		\item[--] For 28 of them there exist a subspace $S_i$ of $\F_2^m$ of dimension $m-3=2$, s.t. $D_aD_b\pi_i=0$ for all $a,b \in S_i$. An example of such a permutation $\pi_i$ and a subspace $S_i$ is given by:
			\begin{equation*}
			\pi_1(y)=
			\begin{bmatrix}
			y_1 \\
			y_2 + y_1 y_2 + y_1 y_4 \\
			y_1 y_2 + y_3 + y_2 y_4 \\ 
			y_2 y_3 + y_4 + y_1 y_4 + y_2 y_4 + y_1 y_5 \\
			y_1 y_2 + y_3 y_4 + y_5 + y_1 y_5
			\end{bmatrix}\quad \mbox{and} \quad S_1=
			\left\langle\scalebox{0.85}{$
				\begin{array}{ccccc}
					0 & 0 & 0 & 1 & 0 \\
					0 & 0 & 0 & 0 & 1 \\
				\end{array}$}
			\right\rangle.
		\end{equation*}    
		\item[--] For the remaining 4 permutations, the maximum dimension of $S_i$ s.t. $D_aD_b\pi_i=0$ for all $a,b \in S_i$ is equal to $(m-2)=3$. An example of such a permutation $\pi_i$ and a subspace $S_i$ is given by:
	\begin{equation*}\label{eq: subspace S with dimS m-2}
		\pi_2(y)=
		\begin{bmatrix}
			y_1\\ y_2 + y_1 y_2 + y_1 y_3\\ y_3 + y_1 y_3 + y_1 y_5\\ y_1 y_2 + y_4 + y_1 y_4\\ y_2 y_3 + y_1 y_4 + y_5 + y_1 y_5
		\end{bmatrix}\quad \mbox{and} \quad S_2=
		\left\langle\scalebox{0.85}{$
			\begin{array}{ccccc}
				0 & 0 & 1 & 0 & 0 \\
				0 & 0 & 0 & 1 & 0 \\
				0 & 0 & 0 & 0 & 1 \\
			\end{array}$}
		\right\rangle.
	\end{equation*}    
\end{itemize}
\end{rem}

\section{Explicit constructions of permutations with $(P_1)$ and $(P_2)$ properties}\label{sec: 4}
The main aim of this section is to specify certain  classes of permutations on $\F_2^m$ satisfying either $(P_1)$ or $(P_2)$ property, and thus to provide constructions of Maiorana-McFarland bent functions with the unique canonical $\mathcal{M}$-subspace $\F_2^{m} \times \{0_m \}$.

\subsection{APN and APN-like permutations}
In the following remark, we indicate that APN permutations have the property~\eqref{eq: P1}, and, hence, can be used for the construction of Maiorana-McFarland bent functions with the unique canonical $\mathcal{M}$-subspace.
\begin{rem}
	Recall that a function $F\colon\F_2^m\to \F_2^m$ is called \textit{almost perfect nonlinear (APN)} if, for all $a\in\F_2^m\setminus\{0_m\},b\in\F_2^m$, the equation $F(x+a)+F(x)=b$ has 0 or 2 solutions $x\in\F_2^m$. Using the notation in~\cite{Li20,Meidl2023DM}, for $n \ge 2$, we define the set of all $2$-dimensional flats in $\F_2^m$ as follows:
	$$	\cB_m=\{ \{x_1,x_2,x_3,  x_4\}  \mid \mbox{$x_1+x_2+x_3+x_4=0_m$ } 
		 \mbox{and $x_1,x_2,x_3,x_4 \in\F_2^m$ are distinct} \}.$$
	
	It is well-known, that a function $F\colon \F_2^m \rightarrow \F_2^m$ is APN if and only if for each $\{x_1,x_2,x_3,x_4\} \in \cB_m$, holds
	$$
	F(x_1)+F(x_2)+F(x_3)+F(x_4) \ne 0_m.
	$$
	Namely, the summation of $F$ over each $2$-dimensional flat is non-vanishing. For a function $F\colon \F_2^m \rightarrow \F_2^m$, define the set of \emph{vanishing flats} with respect to $F$ as
	$$	\VB_{m,F}=\{ \{x_1,x_2, x_3, x_4\} \in \cB_m \mid 
		  F(x_1)+F(x_2)+F(x_3)+F(x_4)=0_m \}.$$
	With this notation, $F$ is APN on $\F_2^m$ if and only if $\VB_{m,F}=\varnothing$. Therefore, any permutation $\pi$ of $\F_2^m$, which is APN,  satisfies the condition~\eqref{eq: P1}. 
	For instance, all power APN functions $x\mapsto x^d$ are permutations of $\F_2^m$ for $m$ odd, as shown by Dobbertin, for the proof we refer to~\cite{Carlet2021}.
\end{rem}
Note that if a function $\pi$ on $\F_2^m$ is quadratic, then $D_{a,b}\pi(y) =const$ for all $a,b\in\F_2^m$. In this way, with the ``vanishing flats'' characterization of APN functions, we deduce the following characterization of quadratic permutations with the~\eqref{eq: P1} property.
\begin{cor}
	A quadratic permutation $\pi$ of $\F_2^m$ has the~\eqref{eq: P1} property if and only if $\pi$ is a quadratic APN permutation of $\F_2^m$.
\end{cor}
\begin{ex}
	Every bent function in $n=6$ variables with the unique $\mathcal{M}$-subspace is equivalent to a bent function of the form $f(x,y)=Tr(x y^3)$, for $x,y\in\F_{2^3}$. In this case, $y\mapsto y^3$ is an APN permutation of $\F_{2^3}$.
\end{ex}

Further, we  show that the following family of quadratic \textit{APN-like permutations}, i.e., non-APN permutations with a small number of vanishing flats (relative to the total number of vanishing flats), have the~\eqref{eq: P2} property. In this way, they can be used for constructing bent functions with the unique $\mathcal{M}$-subspace.
\begin{theo}\cite{Li20}\label{th: coverDO}
	Let $\pi(x)=x^{2^t+1}$ be a function over $\mathbb{F}_{2^m}$ with $(m, t)=s>1$. 
	Then, $\displaystyle\left|\mathcal{V} \mathcal{F}_{m, \pi}\right|=$ $2^{n-2}\left(2^{s-1}-1\right)\cdot\left(2^n-1\right)/3$.
\end{theo}
%
The following characterization of linear structures of the components of permutation  monomials given in \cite{Pascale_Gohar2010} (stated only for binary quadratic case) is useful for our purpose.
\begin{theo}  \cite{Pascale_Gohar2010} \label{th:Pas_Goh}
	Let $\delta \in \F_{2^m}$ and $1 \leq s \leq 2^m-2$ be such that $f(x)=Tr(\delta x^s)$ is not the zero function  on $\F_2^m$. Then, when $wt_H(s)=2$ the function $f$ has a linear structure if and only if the following is true:\\
	(ii): $s=2^j(2^i+1)$, 
	where $0 \leq i, j \leq m-1$, 
$i \not \in \{0,m/2\}$.
	In this case,  $\alpha \in \F_{2^m}$ is a linear structure of $f$ if and only if it satisfies $(\delta^{2^{m-j}} \alpha^{2^i+1})^{2^i-1}+1=0$. More exactly the linear space $\Lambda$ of $f$ is as follows. 
	Denote $\sigma =\gcd(m,2i)$. Then, $\Lambda=\{0\}$ if $\delta$ is not a $(2^i+1)$-th power in $\F_{2^m}$. Otherwise, if $\delta=\beta^{2^j(2^i+1)}$ for some $\beta \in \F_{2^m}$, it holds that $\Lambda=\beta^{-1}\F_{2^\sigma}$.
\end{theo}
\begin{prop}\label{prop: APN-like}
	Let $\pi (y) = y^{2^t+1}$ for $y\in\F_{2^m}$, where $s=\gcd(t, m) = 2$, $m = 2r$ and $r\geq 3$ is odd. Denote by $S$ a vector subspace of $\F_{2^m}$ for which
	$D_aD_b \pi(y) = 0_m$, for any $a,b\in S$.  Then, $\dim(S)\le2$ and permutation $\pi$ has the property~\eqref{eq: P2}.
\end{prop}
\begin{proof}
We first notice that when $\dim(S)=1$ we trivially have that $D_aD_b \pi(y) = 0_m$, since either $a$ or $b$ is zero.  To prove that $\pi$ has the property~\eqref{eq: P2}, let $S$ be a vector subspace of $\F_{2^m}$ for which $D_aD_b \pi(y) = 0_m$, such that $\dim(S)=2$. We will show that there do not exist  linearly independent $u_1, \ldots, u_{m-2} \in \F_2^m$ such that $Tr(u_i D_a \pi)= D_a (Tr(u_i \pi))=0$, for all $a \in S$ and $i=1, \ldots m-2$. Let $u_1, \ldots, u_{m-2}$ be any $m-2$ linearly independent elements in $\F_2^m$. Set $j=0$ and $i=t$ in Theorem \ref{th:Pas_Goh}. Since $m = 2r$, $r$ is odd and $\gcd(t, m) = 2$, we have that $\gcd(2t, m)=2$, i.e., $\sigma=2$ in Theorem \ref{th:Pas_Goh}. From Theorem \ref{th:Pas_Goh}, we deduce that the linear space of $Tr( \delta y^{2^t+1})$ is $\beta ^{-1} \mathbb{F}_{2^2}$, where $\beta$ is such that $\delta = \beta^{2^t+1}$. This means that the linear space of $Tr( u_i y^{2^t+1})$ is $\beta_i^{-1}\mathbb{F}_{2^2}$, where $u_i= \beta_i^{2^t+1}$, for $i=1, \ldots m-2$. Since $u_1, \ldots u_4$ are four linearly independent vectors, then $\beta_1^{-1}$, $\beta_2^{-1}$, $\beta_3^{-1}$,  $\beta_4^{-1}$ are four different nonzero elements, and hence we have that for at least two, w.l.o.g., $u_1$ and $u_2$ the subspaces $\beta_1^{-1} \mathbb{F}_{2^2} $  and $\beta_2^{-1}\mathbb{F}_{2^2} $ are different. The subspace $S$ does not cover both of them, w.l.o.g., assume that it does not cover $\beta_1^{-1} \mathbb{F}_{2^2}$. Let $a \in S\setminus \lbrace 0 \rbrace$ be such that $a \notin \beta_1^{-1} \mathbb{F}_{2^2}$, which  exists since both $S$ and $\beta_1^{-1}\mathbb{F}_{2^2}$ have $4$ elements and $S$ does not cover $\beta_1^{-1}\mathbb{F}_{2^2}$. Then, since $\beta_1^{-1} \mathbb{F}_{2^2}$ is the linear space of $Tr( u_1 y^{2^t+1})$, we have that $D_a(Tr( u_1 y^{2^t+1}))$ is not constant. Since $u_1, \ldots, u_{m-2}$ were arbitrary linearly independent elements from $\F_2^m$, we deduce that there do not exist  linearly independent 
		$u_1, \ldots, u_{m-2} \in \F_2^m$ for which $Tr(u_i D_a \pi)= D_a (Tr(u_i \pi))=0$, for all $a \in S$ and $i=1, \ldots m-2$. That is $\pi$ has the property~\eqref{eq: P2}.
  
    Assume that $\dim(S)=t$, where $3 \leq t \leq m-1$, and assume that there exist $u_1, \ldots, u_{m-t} \in \F_2^m$ such that $Tr(u_i D_a \pi)= D_a (Tr(u_i \pi))=0$, for all $a \in S$ and $i=1, \ldots m-t$.
    From Theorem \ref{th:Pas_Goh}, we have that the linear space of $Tr( u_i y^{2^t+1})$ is $\beta_i^{-1}\mathbb{F}_{2^2}$, where $u_i= \beta_i^{2^t+1}$, for $i=1, \ldots m-t$. Since $\dim(S) \geq 3$, there is an element $a \in S$ such that $a \neq \beta_i^{-1}\mathbb{F}_{2^2}$. This means that $a \in S$ is not in the linear space of $Tr( u_1 y^{2^t+1})$, hence $D_a(Tr( u_1 y^{2^t+1}))$ is not constant, which is a contradiction with our assumption $D_a (Tr(u_1 \pi))=0$.
\end{proof}


\subsection{Piecewise permutations having~\eqref{eq: P1} property}
Now, we provide a secondary construction of permutations with the~\eqref{eq: P1} property. In this way, we obtain infinite families of permutations with the~\eqref{eq: P1} in all dimensions. We also indicate that permutations with the~\eqref{eq: P1} property are not necessarily APN.

	\begin{prop}\label{prop:gensecderlarger}
		Let $\sigma_1$ and $\sigma_2$ be two permutations of $\F_2^m$ such that $D_V\sigma_1 \neq D_V\sigma_2$ for all two dimensional subspaces $V$ of $\F_2^m$. Define the function $\pi \colon \F_2^{m+1} \to \F_2^{m+1}$ by
		\begin{equation*}
			\pi(y,y_{m+1})= \left( \sigma_1(y) + y_{m+1}(\sigma_1(y)+\sigma_2(y)) , y_{m+1} \right) \text{, for all } y \in \F_2^m, y_{m+1} \in \F_2.
		\end{equation*}
		Then the function $\pi$ is a permutation of $\F_2^{m+1}$ such that $D_W\pi \neq 0_{m+1}$ for  all two dimensional subspaces $W$ of $\F_2^{m+1}$.
	\end{prop}
	\begin{proof}
		Since $\pi(y,0)=(\sigma_1(y),0)$ and $\pi(y,1)=(\sigma_2(y),1)$ and since $\sigma_1$ and $\sigma_2$ are permutations, $\pi$ is a permutation as well.
		
		Take two linearly independent vectors $(a,a_{m+1}),(b,b_{m+1})\in \F_2^{m+1}$, where $a,b\in\F_2^{m}$ and $a_{m+1},b_{m+1}\in\F_2$. 
		
		Assume first that $a_{m+1}=b_{m+1}=0$. Then  
		\begin{equation*}
			D_{(a,a_{m+1})}D_{(b,b_{m+1})} \pi(y,y_{m+1})= (D_aD_b \sigma_1(y)+y_{m+1}( D_aD_b \sigma_1(y) +D_aD_b \sigma_2(y),0) 
		\end{equation*}
		Since $(a,a_{m+1})$ and $(b,b_{m+1})$ are linearly independent and $a_{m+1}=b_{m+1}=0$, the vectors $a$ and $b$ are linearly independent. If $D_aD_b \sigma_1(y) \neq 0_m$, then $D_{(a,a_{m+1})}D_{(b,b_{m+1})} \pi(y,0) = (D_aD_b\sigma_1(y),0) \neq 0_{m+1}$, hence $D_{(a,a_{m+1})}D_{(b,b_{m+1})} \pi(y,y_{m+1}) \neq 0_{m+1}$. If $D_aD_b \sigma_1(y) = 0_m$, then, since from the assumption $D_aD_b \sigma_2(y) \neq D_aD_b \sigma_1(y) = 0_m$, we have that $$D_{(a,a_{m+1})}D_{(b,b_{m+1})} \pi(y,1) = ( \sigma_2(y),0) \neq 0_{m+1},$$ hence $D_{(a,a_{m+1})}D_{(b,b_{m+1})} \pi(y,y_{m+1}) \neq 0_{m+1}$. We conclude that in any case, when $a_{m+1}=b_{m+1}=0$, we have $D_{(a,a_{m+1})}D_{(b,b_{m+1})} \pi(y,y_{m+1}) \neq 0_{m+1}$.
		
		Now assume that $a_{m+1}=1$ or $b_{m+1}=1$. W.l.o.g, we assume that $b_{m+1}=1$. Then, since 
		$$D_{(a,a_{m+1})}D_{(b,b_{m+1})} \pi(y,y_{m+1}) =D_{(a+b,a_{m+1}+b_{m+1})}D_{(b,b_{m+1})} \pi(y,y_{m+1}),$$ we can assume that $a_{m+1}=0$. Computing the second-order derivative of $\pi$, we get
		\begin{equation*}
			\begin{split}
				D_{(a,a_{m+1})}&D_{(b,b_{m+1})} \pi(y,y_{m+1})=D_{(b,1)} (D_a\sigma_1(y)+y_{m+1}(D_a\sigma_1(y)+D_{a}\sigma_2(y)),0)\\
				&=(D_aD_b\sigma_1(y) +y_{m+1}(D_{a}D_{b}\sigma_1(y)+D_{a}D_{b}\sigma_2(y))+D_{a}\sigma_1(y+b)+D_{a}\sigma_2(y+b),0),
			\end{split}
		\end{equation*}
		for all $y \in \F_2^m, y_{m+1} \in \F_2$. Setting $y_{m+1}=0$, we have $$D_{(a,a_{m+1})}D_{(b,b_{m+1})} \pi(y,0)=(D_aD_b\sigma_1(y) +D_{a}\sigma_1(y+b)+D_{a}\sigma_2(y+b),0).$$
		If $D_aD_b\sigma_1(y) +D_{a}\sigma_1(y+b)+D_{a}\sigma_2(y+b) \neq 0_m$, we deduce that $D_{(a,a_{m+1})}D_{(b,b_{m+1})} \pi(y,0) \neq 0_{m+1}$, hence $D_{(a,a_{m+1})}D_{(b,b_{m+1})} \pi(y,y_{m+1}) \neq 0_{m+1}$. If however, $D_aD_b\sigma_1(y) +D_{a}\sigma_1(y+b)+D_{a}\sigma_2(y+b) = 0_m$, then we compute 
		$$D_{(a,a_{m+1})}D_{(b,b_{m+1})} \pi(y,1)=(D_{a}D_{b}\sigma_1(y)+D_{a}D_{b}\sigma_2(y),0).$$
		From the assumption  $D_aD_b \sigma_2(y) \neq D_aD_b \sigma_1(y)$, we have $D_aD_b \sigma_2(y) + D_aD_b \sigma_1(y) \neq 0_m$, hence $D_{(a,a_{m+1})}D_{(b,b_{m+1})} \pi(y,1) \neq 0_{m+1}$, and consequently $D_{(a,a_{m+1})}D_{(b,b_{m+1})} \pi(y,y_{m+1}) \neq 0_{m+1}$. We deduce that $D_{(a,a_{m+1})}D_{(b,b_{m+1})} \pi(y,y_{m+1}) \neq 0_{m+1}$, what concludes the proof.
	\end{proof}
	
	\begin{cor}\label{cor:secderlarger}
		Let $\sigma$ be a permutation of $\F_2^m$ such that $D_V\sigma \neq 0_m$ for all two dimensional subspaces $V$ of $\F_2^m$. Define the function $\pi \colon \F_2^{m+1} \to \F_2^{m+1}$ by
		\begin{equation}\label{eq:permext}
			\pi(y,y_{m+1})= \left( y + y_{m+1}(\sigma(y)+y) , y_{m+1} \right) \text{, for all } y \in \F_2^m, y_{m+1} \in \F_2.
		\end{equation}
		Then,  $\pi$ is a permutation of $\F_2^{m+1}$ such that $D_W\pi \neq 0_{m+1}$ for  all two dimensional subspaces $W$ of $\F_2^{m+1}$, thus it satisfies the~\eqref{eq: P1} property.
	\end{cor}
	\begin{proof}
		Set $\sigma_1(y)=y$ and $\sigma_2(y)=\sigma(y)$ for all $y \in \F_2^m$. Then $D_V\sigma_1(y) = 0_m \neq D_V\sigma_2(y)$ for all two dimensional subspaces $V$ of $\F_2^m$. The result then follows from Proposition \ref{prop:gensecderlarger}.
	\end{proof}
	\noindent
	Note that, with the same assumptions as in Corollary \ref{cor:secderlarger}, using Proposition \ref{prop:gensecderlarger} and setting $\sigma_1(y)=\sigma(y)$ and $\sigma_2(y)=y$, we can deduce in the same way that $$\pi'(y,y_{m+1})= \left( \sigma(y) + y_{m+1}(\sigma(y)+y) , y_{m+1} \right)$$ is also a permutation such that $D_W\pi' \neq 0_{m+1}$ for  all two dimensional subspaces $W$ of $\F_2^{m+1}$.

In the following remark, we indicate that APN-ness of permutations $\pi$ on $\F_2^m$ with the~\eqref{eq: P1} property, plays a very important role in the vanishing behaviour of Maiorana-McFarland bent functions $x\cdot\pi(y)$. 
\begin{rem} Let $\sigma$ be a permutation on $\F_2^m$ such that $D_V\sigma \neq 0_m$ for all two dimensional subspaces $V$ of $\F_2^m$. Define the permutation $\pi \colon \F_2^{m+1} \to \F_2^{m+1}$ as in Corollary \ref{cor:secderlarger} by $$\pi(y,y_{m+1})= \left( y + y_{m+1}(\sigma(y)+y) , y_{m+1} \right) \text{, for all } y \in \F_2^m, y_{m+1} \in \F_2.$$
	Clearly, the permutation $\pi$ is not APN, since the last coordinate is linear. Define the function $f\colon \F_2^{2m+2} \to \F_2$ by $$f(x,x_{m+1},y,y_{m+1})=(x,x_{m+1}) \cdot \pi(y,y_{m+1}),$$ for all $x,y \in \F_2^m$ and $x_{m+1},y_{m+1} \in \F_2$. From Corollary \ref{cor:secderlarger} and 
	Theorem \ref{theo: unique P1} we deduce that $\pi$ has the property \eqref{eq: P1}, and $\F_2^{m+1} \times \{ 0_{m+1} \}$ is the unique $\mathcal{M}$-subspace of $f$.
	
	Now, define $a_1=\mathbbm{e}_{m+1} \in\F_{2}^{m+1},a_2=0_{m+1} \in\F_{2}^{m+1}$ and $b_1=\mathbbm{0}_{m+1} \in\F_{2}^{m+1},b_2=(b,0) \in\F_{2}^{m+1}$, where $b$ is a nonzero vector in $\F_2^m$.  From \eqref{eq: secder2}, we have 
	\begin{equation*}
		\begin{split}
			D_{(a_1, a_2)}D_{(b_1, b_2)}f(x,x_{m+1},y,y_{m+1})
			&=(x,x_{m+1})\cdot D_{a_2}D_{b_2}\pi(y,y_{m+1})\\&+  a_1\cdot D_{b_2}\pi((y,y_{m+1})+ a_2)
			+ b_1 \cdot D_{a_2}\pi((y,y_{m+1})+ b_2) \\
			&=\mathbbm{e}_{m+1}\cdot D_{(b,0)}\pi(y,y_{m+1}) \\
			&=\mathbbm{e}_{m+1} \cdot (b+y_{m+1}(D_b \sigma(y) + b),0)\\
			&= 0.
		\end{split}
	\end{equation*}
	However, $\dim(\langle (a_1, a_2),(b_1, b_2) \rangle)=2$, and since $b_2=(b,0) \neq 0_{m+1}$, it is not a subspace of $\F_2^{m+1} \times \{ 0_{m+1} \}$. This means that $D_{a}D_{b}f=0$ vanishes not only on the two-dimensional subspaces $\{a,b\}$ of $\F_2^m\times\{0_m\}$, from what follows that not every permutation $\pi$ with the~\eqref{eq: P1} property defines the bent function $(x,y)\mapsto x\cdot\pi(y)$ with the vanishing behavior as in Corollary~\ref{cor:dim2}. 
\end{rem}
The problem of preserving the~\eqref{eq: P2} property for the class of permutations defined by \eqref{eq:permext} appears to be harder. One can eventually show that the~\eqref{eq: P2} property for $\pi$ is inherited from $\sigma$ for some particular subspaces whereas it remains an open problem to show  that  $\pi$ fully satisfies  the~\eqref{eq: P2} property when $\sigma$ does.
\begin{op}
	Find more constructions of permutations with the~\eqref{eq: P2} property.
\end{op}

\section{Generic construction methods of bent functions outside $\cM^\#$}\label{sec: 5}
In this section, we provide a theoretical analysis of possible $\mathcal{M}$-subspaces of the bent 4-concatenation $f=f_1||f_2||f_3||f_4\in\mathcal{B}_{n+2}$. Based on this analysis, we consequently provide two generic methods of constructing bent functions outside $\cM^\#$ for even $n \geq 8$. Our first approach is based on the concatenation of  bent functions $f_1,f_2,f_3,f_4\in\mathcal{B}_{n}$ that \textit{do not share any $\mathcal{M}$-subspace of dimension $n/2-1$}, i.e, $\bigcap_{i=1}^{4}\mathcal{MS}_{n/2-1}(f_i)=\varnothing$.  Our second approach is based on the concatenation of  bent functions $f_1,f_2,f_3,f_4\in\mathcal{B}_{n}$ that \textit{share a unique $\mathcal{M}$-subspace of dimension $n/2$}, i.e, $|\bigcap_{i=1}^{4}\mathcal{MS}_{n/2}(f_i)|=1$. Finally, we provide an algorithm for checking the membership in the completed partial spread class $\mathcal{PS}^\#$, and show that with our approaches it is possible to construct inequivalent bent functions in $n=8$ outside $\cM^\#\cup\mathcal{PS}^\#$.
\subsection{Possible $\mathcal{M}$-subspaces of the bent 4-concatenation}\label{sub: 5.1}
The following result is crucial in understanding the structural properties of bent functions in $\cM^\#$ in terms of 4-concatenation. Notice that when considering $f=f_1||f_2||f_3||f_4 $ we do not assume neither that $f_i$ are bent nor that $f_i$ share the same unique $\mathcal{M}$-subspace.
	\begin{prop} \label{prop:commonsubspace} Let $f_1, \ldots,f_4$ be four Boolean functions in $n$ variables, not necessarily bent, such that $f=f_1||f_2||f_3||f_4 \in\B_{n+2}$ is a bent function in $\cM^{\#}$. Let $W$ be an $\mathcal{M}$-subspace of $f$ of dimension $(\frac{n}{2}+1)$. Then, there is an $(\frac{n}{2}-1)$-dimensional subspace $V$ of $\F_2^n$ such that:
	\begin{itemize}
		\item[1)] $V \times \{(0,0)\}$ is a subspace of $W$,
		\item[2)] $V$ is an $\mathcal{M}$-subspace of $f_i$ for all $i=1, \ldots ,4$.
	\end{itemize}	
	\end{prop}
	\begin{proof} Let $W$ be an $\mathcal{M}$-subspace of $f$ of dimension $(\frac{n}{2}+1)$ (we know that it exists since $f$ is in $\cM^{\#}$). We have 
		\begin{equation*}
		\dim ( W \cap (\F_2^n \times \{(0,0) \}))= \dim (W) + \dim (\F_2^n \times \{(0,0) \}) - \dim(\langle W , \F_2^n \times \{(0,0) \} \rangle).
		\end{equation*}
		Because $\dim ( W + (\F_2^n \times \{(0,0) \})) \leq n+2$, we have 
		\begin{equation*}
		\dim ( W \cap (\F_2^n \times \{(0,0) \})) \geq (\frac{n}{2}+1) + n -(n+2) = \frac{n}{2}-1.
		\end{equation*}
		Hence, there is an $(\frac{n}{2}-1)$-dimensional subspace $V$ of $\F_2^n$ such that $V \times \{ (0,0) \}$ is a subspace of $W$. Let $a$ and $b$ be two arbitrary vectors from $V$. Then $(a,0,0)$ and $(b,0,0)$ are in $W$, so $D_{(a,0,0)}D_{(b,0,0)}f=0$. Using~\eqref{eq:2ndderiv_conc correct}, we compute:
		\begin{eqnarray}\label{eq:computDaDB}
		D_{(a,0,0)}D_{(b,0,0)}f(x,z_1,z_2) &=& D_aD_bf_1(x)+z_1(D_aD_b(f_1+f_2)(x))+z_2(D_aD_b(f_1+f_3)(x)) \nonumber \\
		& &+ z_1z_2(D_aD_b(f_1+f_2+f_3+f_4)(x))=0,
		\end{eqnarray}
		for all $(x,z_1,z_2) \in \F_2^{n+2}$. From this, we deduce that
		\begin{equation} \label{eq:conditions}
		D_aD_bf_1(x)=D_aD_b(f_1+f_2)(x)=D_aD_b(f_1+f_3)(x)=D_aD_b(f_1+f_2+f_3+f_4)(x)=0,
		\end{equation}
		for all $x \in \F_2^n$, and consequently, that $D_aD_bf_1=D_aD_bf_2=D_aD_bf_3=D_aD_bf_4=0$. Since $a$ and $b$ were two arbitrary elements from $V$ this completes the proof.
	\end{proof}

	As a special case of concatenating four bent functions $f_i \in \B_n$ in $\cM$, 
that share the same unique vanishing subspace $V=\F_2^m \times \{0_m\}$, we have the following important result that describes the form of $\mathcal{M}$-subspaces for $f=f_1||f_2||f_3||f_4$.
\begin{prop} \label{prop:sharinguniqueM} Let $f_1, \ldots,f_4 \in \B_n$, with $n=2m$, all belong to the $\cM$ class and additionally assume that the only $n/2$-dimensional subspace $U$ of $\F_2^n$  for which $D_aD_b f_i=0$ for all $a,b  \in U$, is   given by $U=\F_2^m \times \{0_m\}$. Then, the only possible $(n/2+1)$-dimensional $\mathcal{M}$-subspaces $\{W\}$ for $f=f_1||f_2||f_3||f_4$ are of the following form:  
	\begin{enumerate}[i)]
		\item    $W= \langle U \times (0,0),(a,b,c_1,c_2)\rangle $,  where  $c_1,c_2 \in \F_2$ and $(c_1,c_2)\neq 0_2$; or
		$W=\langle V \times (0,0), (a,b,c_1,c_2),  (e,f,d_1,d_2)\rangle $, where $V \subset U$ with  $\dim(V)=n/2-1$, 
		$(c_1,c_2)\neq 0_2, (d_1,d_2)\neq 0_2, (c_1,c_2)\neq (d_1,d_2)$.
		\item   $W =\langle U' \times (0,0), (a,b,c_1,c_2),  (e,f,d_1,d_2)\rangle $, where $\dim(U')=n/2-1$ and $U'\not\subset U$,  $(c_1,c_2)\neq 0_2, (d_1,d_2)\neq 0_2, (c_1,c_2)\neq (d_1,d_2)$.		
	\end{enumerate}
\end{prop}
\begin{proof}
	By Proposition~\ref{prop:commonsubspace}, if  $f \in \cM^\#$ then  any $(n/2+1)$-dimensional  $\mathcal{M}$-subspace $W$ of $f$ contains an $(n/2-1)$-dimensional (shared) subspace $V$ of $\F_2^n$ such that $D_aD_bf_i=0$, for all $a,b \in V$ and  $i=1, \ldots ,4$. 
	By assumption, this $(n/2-1)$-dimensional subspace $V$ of $\F_2^n$ such that $D_aD_bf_i=0$, for all $a,b \in V$ and  $i=1, \ldots ,4$, is either  a subspace of $U=\F_2^m \times \{0_m\}$ or alternatively $V \not \subset U$.  Furthermore,  by Proposition~\ref{prop:commonsubspace}, if $f \in \cM^\#$ then 
	$V \times \{ (0,0) \}$ is a vanishing subspace of $\F_2^{n+2}$ (of dimension $n/2-1$) for   $f$. 
	Notice that since $\dim(W)=n/2+1$ and $V \times (0,0) \subset W$, then 
	$$d=\dim(\{(a,b,0_2)\in\vF{n/2}\times\vF{n/2} \times \F_2^2 \colon (a,b,c_1,c_2)\in W\})\geq n/2-1. $$
	However, we also have  that  $d \leq n/2$ since any bent function on $\F_2^n$  cannot have an $\mathcal{M}$-subspace of dimension larger than $n/2$, which can be deduced from \cite[Proposition 8.33]{Carlet2021} and is explicitly stated in \cite[Result 1.35]{Polujan2020}. 

	Thus, there are two cases to consider. 
	
	\noindent $a)$ {\bf The case $V \subset U$:} 
	This  implies that we have two situations here.
	When  $d=n/2$, that is $V$ is extended to $U$, so that $W^{(1)}= \langle U \times (0,0), (a,b,c_1,c_2)\rangle $ is an $\mathcal{M}$-subspace of $f$. 
	
	When  $d=n/2-1$, then we have 
	$W^{(2)}=\langle V \times (0,0), (a,b,c_1,c_2),  (e,f,d_1,d_2)\rangle $  is an $\mathcal{M}$-subspace of $f$, where $V \subset U$ with  $\dim(V)=n/2-1$.   
	Assuming  that $(c_1,c_2)= 0_2$  or $ (d_1,d_2)= 0_2$, would contradict that $d=n/2-1$ and lead to $W^{(2)}= W^{(1)}$. 
	Similarly, one can  deduce  $  (c_1,c_2)\neq (d_1,d_2)$ as otherwise we would get $d=n/2$.
	It is obvious that $ W^{(1)}\neq W^{(2)}$.
	
	
	\noindent $b)$ {\bf The case $V \not  \subset U$:} 
	We have only the case $d=n/2-1$ since by  assumption $f_1, \ldots,f_4 \in \B_n$ have  only  $(n/2-1)$-dimensional subspaces $V$ of $\F_2^n$  for which $D_aD_b f_i=0$ for all $a,b  \in V$. Hence,  we have $W^{(3)}=\langle V \times (0,0), (a,b,c_1,c_2),   (e,f,d_1,d_2)\rangle $, where $\dim(V)=n/2-1$. 		Without loss of generality, we assume $(c_1,c_2)\neq 0_2$,
	then $d=n/2$ which contradicts that $d=n/2-1$.    Similarly, we know  $  (d_1,d_2)\neq 0_2, (c_1,c_2)\neq (d_1,d_2)$.
	Hence, we have  $(c_1,c_2)\neq  0_2, (d_1,d_2)\neq 0_2, (c_1,c_2)\neq (d_1,d_2)$.
	It is obvious that $W^{(3)}\neq W^{(1)}$. Now we prove that  $W^{(3)}\neq W^{(2)}$.
	Since $V\not  \subset U$, we have
	$$\{(a,b,0_2) \colon (a,b,c_1,c_2)\in W^{(3)}\}\neq \{(a,b,0_2) \colon (a,b,c_1,c_2)\in W^{(2)}\},$$
	which confirms the claim.
\end{proof}
\paragraph{An algorithm for checking the membership in the $\mathcal{PS}^\#$ class.} 
Recall that a partial spread of order $s$ in $\mathbb{F}_{2}^{n}$ with $n=2m$ is a set of $s$ vector subspaces $U_{1}, \ldots, U_{s}$ of $\mathbb{F}_{2}^{n}$ of dimension $m$ each, such that $U_{i} \cap U_{j}=\{0_n \}$ for all $i\neq j$. The partial spread of order $s=2^{m}+1$ in $\mathbb{F}_{2}^{n}$ with $n=2m$ is called a spread.

In the following, we denote by $\mathbbm{1}_{U}\colon \mathbb{F}_{2}^{n} \to \mathbb{F}_{2}$ the \emph{indicator function} of $U\subseteq\F_2^n$, i.e., $\mathbbm{1}_{U}(x)=1$ if $x\in U$, and $0$ otherwise. The \emph{partial spread class} $\mathcal{PS}$ of bent functions on $\F_2^n$ is the union of the following two classes~\cite{Dillon}: the $\mathcal{PS}^{+}$ \emph{class} is the set of Boolean bent functions of the form $f(x)=\sum_{i=1}^{2^{m-1}+1} \mathbbm{1}_{U_i}(x)$; the $\mathcal{PS}^{-}$ \emph{class} is the set of Boolean bent functions of the form $f(x)=\sum_{i=1}^{2^{m-1}} \mathbbm{1}_{U^*_i}(x)$, where $U^*_i:=U_i\setminus \{ 0 \}$.
The \emph{Desarguesian partial spread} class $\mathcal{PS}_{ap}\subset \mathcal{PS}^-$ is the set of Boolean bent functions $f$ on $\F_{2^{m}}\times \F_{2^{m}}$ of the form $f\colon (x,y)\in\F_{2^{m}}\times \F_{2^{m}}\mapsto h\left(x/y\right)$, where $\frac{x}{0}=0$, for all $x \in \mathbb{F}_{2^k}$ and $h\colon\F_{2^{k}}\rightarrow\F_2$ is a balanced Boolean function with $h(0)=0$. 

The property of a bent function to be a member of the partial spread class is not invariant under equivalence. If $f$ is partial spread function on $\F_2^n$, i.e., $f(x)=\sum_{i=1}^{s} \mathbbm{1}_{U_i}(x)$ for a partial spread $\{ U_{1}, \ldots, U_{s} \}$ of order $s$ in $\F_2^n$, then for an invertible $n\times n$-matrix $A$, the function $g\colon x \in\F_2^n\mapsto f(xA)$ is a partial spread function as well, since $g(x)=\sum_{i=1}^{s} \mathbbm{1}_{U_iA^{-1}}(x)$ for the partial spread $\{ U_{1}A^{-1}, \ldots, U_{s}A^{-1} \}$. However, translations of the input $x\mapsto x +  b$ for $b\in\F_2^n$ and additions of affine functions $l$ on $\F_2^n$ to the output of a partial spread function $f$ on $\F_2^n$ may lead to functions $g\colon x \mapsto f(x +  b)$ and $h\colon x \mapsto f(x) +  l(x)$ on $\F_2^n$, respectively, which do not belong to the partial spread class $\mathcal{PS}$. In Algorithm~\ref{algorithm: Membership in the partial spread class}, we describe how to check computationally the membership of a given bent function $f$ on $\F_2^n$ in the $\mathcal{PS}$ class.
\begin{algorithm}[H]
	\caption{Membership in the partial spread class $\mathcal{PS}$}
	\label{algorithm: Membership in the partial spread class}
	\begin{algorithmic}[1]
		\Require Bent function $f\in\mathcal{B}_n$.
		\Ensure True, $f$ is a partial spread function and false, otherwise.
		\If{$f(0)=1$} \Comment{The case $\mathcal{PS}^{+}$}
		\State \textbf{Assign} $s:=2^{n/2-1}+1$ \mbox{and} $V:=\operatorname{supp}(f)$ (the support of $f$). 
		\Else \Comment{The case $\mathcal{PS}^{-}$}
		\State \textbf{Assign} $s:=2^{n/2-1}$\quad\quad \mbox{and} $V:=\operatorname{supp}(f)\bigcup \{0_n\}$.
		\EndIf
		\State \textbf{Construct} the graph $G=(V,E)$, for which the relation between vertices in $V$ and edges in $E$ is determined by the incidence matrix $\left[ f( x + y)\right]_{x,y\in V}$.
		\State \textbf{Find} the set $S$ of cliques of the size $2^{n/2}$ in  $G$.
		\State \textbf{Construct} the set $V'$ of cliques in $S$, whose elements form an $n/2$-dimensional vector space.
		\If{$|V'|<k$} 
		\State \textbf{Return} false. 
		\EndIf
		\State \textbf{Construct} the graph $G'=(V',E')$, for which the relation between vertices in $V'$ and edges in $E'$ is determined by the incidence matrix $(a_{i,j})$, where $a_{i,j}=1$, if for $U_i,U_j \in S$ holds $U_i \cap U_j = \{0_n\}$, and 0 otherwise.
		\State \textbf{Return} true, $f$ is a partial spread function, if the graph $G'$ contains a clique of size $k$, and false otherwise.  
	\end{algorithmic}
\end{algorithm}
\begin{rem} Note that, it is possible to establish with Algorithm~\ref{algorithm: Membership in the partial spread class} whether a bent function $f\in\mathcal{B}_n$ belongs to the completed partial spread class $\mathcal{PS}^\#$. If for a vector $b\in \F_2^n$ and an affine function $l$ on $\F_2^n$ the function $g\colon x\mapsto f(x +  b) +  l(x)$ on $\F_2^n$ is a member of the $\mathcal{PS}$ class, we have $f\in\mathcal{PS}^\#$, otherwise $f\notin\mathcal{PS}^\#$.
\end{rem}	

\subsection{Concatenating bent functions on $\F_2^n$ that do not share  any $\mathcal{M}$-subspace of dimension $n/2-1$}\label{sub: 5.2}
With this result,  we derive the following generic construction method of bent functions outside the $\mathcal{MM}^\#$ class.

\begin{theo} \label{th:general_constr} Let $f_1, \ldots,f_4\in\mathcal{B}_n$ be four Boolean functions, not necessarily bent, such that $f=f_1||f_2||f_3||f_4 \in\B_{n+2}$ is a bent function. Assume that there is no 
	 $(\frac{n}{2}-1)$-dimensional subspace $V$ of $\F_2^n$  such that $D_aD_bf_i=0$, for all $a,b \in V$ and all $i\in \{1,\ldots,4\}$. Then, $f \in\B_{n+2}$ is a bent function outside $\cM^\#$.
\end{theo}
\begin{proof}
	The result  is a direct consequence of Proposition \ref{prop:commonsubspace}.
\end{proof}
\begin{ex}\label{ex: Outside everything 1}
 Let $\pi$ be a quadratic APN permutation of $\F_2^3$, which, in turn, has the~\eqref{eq: P1} property:
 \begin{equation}\label{eq: APN perm F23}
 	\pi(y_1,y_2,y_3)=\begin{bmatrix}  y_{2} y_3 + y_{1} + y_{2} + y_3 \\ y_{1} y_{2} + y_{1} y_3 +y_{2}  \\ y_{1} y_{2} + y_3
 	\end{bmatrix}.
 \end{equation}
 Define four bent functions $f_1, \ldots ,f_4\in\mathcal{B}_6$, which all belong to $\cM^\#$, as follows: 
\begin{equation}\label{eq: D0 decomposition}
	\begin{aligned}
		f_1(x,y)&=x\cdot y+\delta_0(x), &
		f_2(x,y)&=x\cdot \pi(y)+\delta_0(x),\\
		f_3(x,y)&=x\cdot y, &
		f_4(x,y)&=x\cdot \pi(y)+1.
	\end{aligned}
\end{equation}
One can check that for defined in~\eqref{eq: D0 decomposition} bent functions, the dual bent condition is satisfied. In this way, we have that $f=f_1||f_2||f_3||f_4\in\mathcal{B}_8$ is bent. Its ANF is given by
\begin{equation}\label{eq: Outside everything 1 ANF}
	\begin{split}
		f(z)=&1 + z_1 + z_2 + z_1 z_2 + z_3 + z_1 z_3 + z_2 z_3 + z_1 z_2 z_3 + z_3 z_4 + z_1 z_5 + z_2 z_6 + z_7 + \\ & z_1 z_7 + z_2 z_7 + z_1 z_2 z_7 + z_3 z_7 + z_1 z_3 z_7 + z_2 z_3 z_7 + z_1 z_2 z_3 z_7 + z_1 z_4 z_8 + z_2 z_4 z_5 z_8 + \\ & z_1 z_6 z_8 + z_1 z_4 z_6 z_8 + z_2 z_5 z_6 z_8 + z_3 z_5 z_6 z_8 + z_7 z_8.
	\end{split}
\end{equation}
 Finally, we confirm that the functions $f_1,f_2,f_3,f_4$ satisfy the conditions of Theorem~\ref{th:general_constr}. Due to the APN-ness of $\pi$, we have that $D_aD_bf_4=0$ if and only if two-dimensional subspace $\{a,b\}$ is a subspace of $S=\F_2^3\times\{0_3\}$. On the other hand,  $D_aD_bf_1\neq0$ for any two dimensional subspace $\{a,b\}$ of $S=\F_2^3\times\{0_3\}$. In this way, we conclude that $f\notin\cM^\#$. Using Algorithm~\ref{algorithm: Membership in the partial spread class}, we also confirm that $f\notin\mathcal{PS}^\#$. In this way, we have that $f\notin(\cM^\#\cup\mathcal{PS}^\#)$. 
\end{ex}
Now, we provide one generic method of specifying $f=f_1||f_2||f_3||f_4$ outside $\cM^\#$, where $f_i$ are bent functions within or outside  $\cM^\#$. The dual bent condition $f_1^*  + f_2^*  +  f_3^*  + f_4^* =1$ can be  satisfied if we simply select, e.g., $f_1=f_2$ and $f_4=1 + f_3$, where $f_i \in \B_n$ are bent. Then, according to Theorem \ref{th:general_constr}, it is enough to ensure that $f_1$ and $f_3$ do not share any $\mathcal{M}$-subspace of dimension $n/2-1$.

\begin{theo}\label{th: two MM with a unique M-subspace} Let $\pi$ be a permutation of $\F_2^m$ having  the property~\eqref{eq: P1}. Let $\sigma$ a permutation of $\F_2^m$, such that there is no $(m-2)$-dimensional subspace $S$ of $\F_2^m$ for which $D_{a}D_{b} \sigma=0$ for all $a,b \in S$.
Let $h_1,h_2\in\mathcal{B}_m$ be arbitrary Boolean functions. Let $f_i\in\mathcal{B}_{2m}$, $i=1,\ldots,4$ be the functions defined by
\begin{equation}\label{eq: bent4 change variables}
    \begin{split}
        f_1(x,y)&=f_2(x,y)=x \cdot \pi(y) + h_1(y),\\
        f_3(x,y)&=f_4(x,y)+1= y \cdot \sigma(x) + h_2(x)
    \end{split}
\end{equation}
for all $x,y \in \F_2^{m}$. Then $f=f_1||f_2||f_3||f_4 \in\B_{2m+2}$ is a bent function outside the $\cM^{\#}$ class.
\end{theo}
\begin{proof}
Assume that $f$ is in the $\cM^{\#}$ class. From Proposition~\ref{prop:commonsubspace}, there exists an $(m-1)$-dimensional subspace $V$ of $\F_2^{2m}$ such that $D_aD_bf_i=0$, for all $a,b \in V$; $i=1, \ldots ,4.$ Define the mapping $L: V \to \F_2^m$ by $L(x,y)=y$, for all $(x,y) \in \F_2^{2m}$. Since $D_aD_bf_1=0$ for all $a,b \in V$, from the proof of Theorem \ref{theo: unique P1} we deduce that $\dim(Im(L)) \leq 1$. From the rank-nullity theorem, we have that $\dim(Ker(L)) \geq m-2$. For $a=(a_1,a_2)$, $b=(b_1,b_2)$ in $Ker(L)$ we have $a_2=b_2=0_m$, and since $Ker(L) \subseteq V$ so $D_aD_bf_3=0$, we get 
$$y \cdot D_{a_1}D_{b_1} \sigma(x) + D_{a_1}D_{b_1} h_2(x) =0, \text{ for all } x,y \in \F_2^m.$$
Consequently, $D_{a_1}D_{b_1} \sigma=0$. Since $\dim(Ker(L)) \geq m-2$, this means that there is a subspace $S$ of $\F_2^m$ of dimension $m-2$ such that $D_{a_1}D_{b_1} \sigma=0$ for all $a_1,b_1 \in S$. However, this is in contradiction with the assumption about $\sigma$. Hence $f$ is outside of the $\cM^{\#}$ class.
\end{proof}

With this result, we can now demonstrate how one can construct bent functions in 8 variables outside $\mathcal{MM}^\#$ class from four bent functions in 6 variables in $\mathcal{MM}^\#$. We emphasize that this is the first  attempt in the literature towards  our   better understanding of the origin of bent functions. 

\begin{ex}\label{ex: Outside everything 2}
    Let $\pi$ be the APN permutation defined in~\eqref{eq: APN perm F23} and $\sigma$ be another APN permutation of $\F_2^3$, defined by the algebraic normal form in the following way:
    \begin{equation*}
        \sigma(x)=
        \begin{bmatrix}
        x_1 + x_2 + x_3 + x_2 x_3\\ x_2 + x_3 + x_1 x_3\\ x_2 + x_1 x_2 + x_1 x_3
        \end{bmatrix}.
    \end{equation*}
    Let $h_1,h_2\in\mathcal{B}_3$ be arbitrary Boolean functions. Define four bent functions $f_i\in\mathcal{B}_{6}$ for $i=1,2,3,4$ as in~\eqref{eq: bent4 change variables}, which all belong to $\cM^\#$. Then, the function $f=f_1||f_2||f_3||f_4 \in\B_{8}$ is a bent function outside the $\cM^{\#}$ class by Theorem~\ref{th: two MM with a unique M-subspace} (independently on the choice of $h_1$ and $h_2$).  Now, set $h_1(y)=y_1 y_2 y_3 + y_1 y_2 + y_1 y_3 + y_2 y_3 + y_1 + y_2 + y_3$ and $h_2(y)=y_1 y_2 y_3 + y_1 y_3 + y_2 y_3 + 1$. Then, the algebraic normal form of $f=f_1||f_2||f_3||f_4$ is given as follows:
    \begin{equation}\label{eq: Outside everything 2 ANF}
    	\begin{split}
    		f(z)&= z_4 + z_1 z_4 + z_5 + z_1 z_5 + z_2 z_5 + z_4 z_5 + z_2 z_4 z_5 + z_3 z_4 z_5 + z_6 + z_1 z_6 + z_3 z_6 \\&+ z_4 z_6 + z_2 z_4 z_6 + z_5 z_6 + z_1 z_5 z_6 + z_4 z_5 z_6 + z_1 z_3 z_7 + z_2 z_3 z_7 + z_1 z_2 z_3 z_7 \\&+ z_4 z_7 + z_2 z_4 z_7 + z_3 z_4 z_7 + 
    		z_2 z_3 z_4 z_7 + z_5 z_7 + z_1 z_5 z_7 + z_1 z_2 z_5 z_7 + z_1 z_3 z_5 z_7 \\&+ 
    		z_4 z_5 z_7 + z_2 z_4 z_5 z_7 + z_3 z_4 z_5 z_7 + z_6 z_7 + z_1 z_6 z_7 + 
    		z_1 z_2 z_6 z_7 + z_4 z_6 z_7 + z_2 z_4 z_6 z_7 \\& + z_5 z_6 z_7 + z_1 z_5 z_6 z_7 + 
    		z_4 z_5 z_6 z_7 + z_7 z_8.\\
    	\end{split}
    \end{equation}
	Using Algorithm~\ref{algorithm: Membership in the partial spread class}, we confirm that $f\notin\mathcal{PS}^\#$, and, hence, $f\notin(\cM^\#\cup\mathcal{PS}^\#)$.
\end{ex}
\begin{rem}\label{rem:sharing}
	It is important to notice that the condition that any $(\frac{n}{2}-1)$-dimensional  $\mathcal{M}$-subspace $V$ is not shared between $f_i$ in Theorem \ref{th:general_constr} is only sufficient, and there exist functions $f_i$ that do share the unique canonical $\mathcal{M}$-subspace $V=\F_2^{n/2} \times \{0_{n/2}\}$  even though 
	$f=f_1||f_2||f_3||f_4$ is outside $\cM^\#$, which is discussed in  Section   \ref{sec:sharingMsubspace}.
\end{rem}

We notice that bent functions on $\F_2^n$ outside $\cM^\#$ do  not  admit $n/2$-dimensional vanishing subspaces, and furthermore it was observed 
in~\cite{Bent_Decomp2022} that  many instances of  bent functions in $\mathcal{PS}\setminus \cM^\#$  only have vanishing  subspaces of dimension
less than $n/2-1$.

\begin{cor}\label{cor:mixingtMMandoutsideMM} Let $f_1=f_2$ be   two arbitrary  bent functions on $\F_2^n$ in $\cM^\#$ and define $f_4=1+f_3$ on $\F_2^n$ where $f_3 \not\in  \cM^\#$ and it does not admit $\mathcal{M}$-subspaces of dimension  larger than   $n/2-2$.  Then, $f=f_1||f_2||f_3||f_4 \in \B_{n+2}$ is a  bent function outside $\cM^\#$. 
\end{cor}
\begin{op}
The non-sharing property provides a theoretical framework for bent 4-concatenation, however finding such $f_i$ (also satisfying the dual bent condition) appears to be difficult. We leave as an open problem a specification of such quadruples in a generic manner.
\end{op}
	\subsection{Concatenating bent functions that share a unique $\mathcal{M}$-subspace of dimension $n/2$} \label{sec:sharingMsubspace}
	Proposition~\ref{prop:sharinguniqueM} provides the possibility to analyze the class exclusion from $\cM^\#$ by only considering the subspaces $W$ of dimension $n/2+1$ of the above form. In particular, this general case is not covered by Proposition~\ref{prop:commonsubspace}, since $f_i$ share the unique $\mathcal{M}$-subspace $U=\F_2^m \times \{0_m\}$. The analysis can be divided into two cases, namely considering the case that the only $(n/2-1)$-dimensional vanishing subspace $U'$ for all $f_i$ is such that $U' \subset U$ or alternatively $U' \not \subset U$. The main problem in this analysis is the fact that $f_1+f_2$, $f_1 + f_3$ or $f_1+f_2+f_3+f_4$ are not in general bent functions and therefore the analysis of second-order derivatives in~\eqref{eq:2ndderiv_conc correct} becomes harder. 
	\begin{theo} \label{th:sharingcommonsubspace}Let $f_1, \ldots,f_4$ be four bent  functions on $\F_2^n$, with $n=2m$, satisfying the following conditions:
	\begin{enumerate}
		\item $f_1, \ldots,f_4$ belong to $\cM^\#$ and  share a unique  $\mathcal{M}$-subspace of dimension $m$;
		\item  $f=f_1||f_2||f_3||f_4 \in\B_{n+2}$ is a bent function;
	\end{enumerate}
	Let $V$ be an $(\frac{n}{2}-1)$-dimensional subspace  of $\F_2^n$  such that $D_aD_bf_i=0$, for all $a,b \in V$; $i=1, \ldots ,4.$
	If for any $v\in  \F_2^n $ and  any such $V \subset \F_2^n$, there exist $u^{(1)},u^{(2)},u^{(3)}\in V $ such that the following three conditions hold simultaneously 
	\begin{enumerate}
		\item[1.] $D_{u^{(1)}}f_1(x)+D_{u^{(1)}}f_2(x+v)\neq 0,~\textit{or}~D_{u^{(1)}}f_3(x)+D_{u^{(1)}}f_4(x+v)\neq 0,$
		\item[2.] $D_{u^{(2)}}f_1(x)+D_{u^{(2)}}f_3(x+v)\neq 0,~\textit{or}~D_{u^{(2)}}f_2(x)+D_{u^{(2)}}f_4(x+v)\neq 0,$
		\item[3.] $D_{u^{(3)}}f_2(x)+D_{u^{(3)}}f_3(x+v)\neq 0,~\textit{or}~D_{u^{(3)}}f_1(x)+D_{u^{(3)}}f_4(x+v)\neq 0,$
	\end{enumerate}
	then $f$ is outside $\cM^\#$.
\end{theo}
\begin{proof}
W.l.o.g., we assume  that the unique $\mathcal{M}$-subspace shared between $f_i$  is $U=\F_2^m \times \{0\}$.	Let  $\{W\}$  be $(n/2+1)$-dimensional subspaces of  $\F_2^{n+2}$.
	We prove that $f$ does not belong to $\cM^\#$ by using  Lemma \ref{lem M-M second}.  We need to show  that, for any $W$, there exist two vectors $ (u,c_1,c_2),(v,d_1,d_2)\in W$  such that 
	$ D_{(u,c_1,c_2)}D_{(v,d_1,d_2)}f\neq 0$.
	
	From Proposition \ref{prop:sharinguniqueM}, if $W$ is an  $(n/2+1)$-dimensional vanishing subspaces of $f$ then 
	$W= \langle U \times (0,0),(a,b,c_1,c_2)\rangle $,  where  $c_1,c_2 \in \F_2,a,b\in \F_2^{n/2}$ and $(c_1,c_2)\neq 0_2$; or
	$W=\langle V \times (0,0), (a,b,c_1,c_2), (e,f,d_1,d_2)\rangle $, where   $\dim(V)=n/2-1$ and $ a,b,e,f\in \F_2^{n/2}$, 
	$(c_1,c_2)\neq 0_2, (d_1,d_2)\neq 0_2, (c_1,c_2)\neq (d_1,d_2)$.
	In addition, we know 
	$$W= \langle U \times (0,0),(a,b,c_1,c_2)\rangle= \langle V \times (0,0),(a,b,c_1,c_2), (e,f,0,0)\rangle,$$ when $V\subset U, (e,f)\in U\setminus V$  (where $\dim(V)=n/2-1$).
	Hence, if we prove that for any $(v,d_1,d_2)\in W$ there always exists one vector $(u,0,0)\in W$ such that  $ D_{(u,0,0)}D_{(v,d_1,d_2)}f\neq 0$ where $(d_1,d_2)\neq 0_2$, then $f$ is outside $\cM^\#$. In order to show it, consider the following three cases. \ \\
	\noindent\textbf{Case 1}. Let $(d_1,d_2)=(0,1)$. From Equation~\eqref{eq:2ndderiv_conc correct}, we have that 
	\begin{equation}\label{equ the5.10 2 c1}
		\begin{split}
			D_{(u,0,0)}D_{(v,d_1,d_2)}f(x,y_1,y_2)=&D_uf_{12}(x+v)+y_1 D_u f_{1234}(x+v)\\
			=&(y_1+1)(D_uf_{12}(x+v))+y_1D_u f_{34}(x+v)\\
			=&(y_1+1)(D_uf_1(x)+D_uf_2(x+v))\\+&y_1(D_uf_3(x)+D_uf_4(x+v)).
		\end{split}
	\end{equation}

	\noindent Since for any $v\in  \F_2^n $ and any $V$, there exist $u^{(1)}\in V $ such that  $D_{u^{(1)}}f_1(x)+D_{u^{(1)}}f_2(x+v)\neq 0,~\textit{or}~D_{u^{(1)}}f_3(x)+D_{u^{(1)}}f_4(x+v)\neq 0$, from (\ref{equ the5.10 2 c1}), we have $$D_{(u^{(1)},0,0)}D_{(v,d_1,d_2)}f(x,y_1,y_2)\neq 0.$$
		\noindent\textbf{Case 2}. Let $(d_1,d_2)=(1,0)$. From Equation~\eqref{eq:2ndderiv_conc correct}, we have that 
	\begin{equation}\label{equ the5.10 2 c2}
		\begin{split}
			D_{(u,0,0)}D_{(v,d_1,d_2)}f(x,y_1,y_2)=&D_uf_{13}(x+v)+y_2 D_u f_{1234}(x+v)\\
			=&(y_2+1)(D_uf_{13}(x+v))+y_2D_u f_{24}(x+v)\\
			=&(y_2+1)(D_uf_1(x)+D_uf_3(x+v))\\+&y_2(D_uf_2(x)+D_uf_4(x+v)).
		\end{split}
	\end{equation}
	
	\noindent Since for any $v\in  \F_2^n $ and any $V$, there exist $u^{(2)}\in V $ such that  $D_{u^{(2)}}f_1(x)+D_{u^{(2)}}f_3(x+v)\neq 0,~\textit{or}~D_{u^{(2)}}f_2(x)+D_{u^{(2)}}f_4(x+v)\neq 0$, from (\ref{equ the5.10 2 c2}), we have $$D_{(u^{(2)},0,0)}D_{(v,d_1,d_2)}f(x,y_1,y_2)\neq 0.$$
	
		\noindent\textbf{Case 3}. Let $(d_1,d_2)=(1,1)$. From Equation~\eqref{eq:2ndderiv_conc correct}, we have that 
	\begin{equation}\label{equ the5.10 2 c3}
		\begin{split}
			D_{(u,0,0)}D_{(v,d_1,d_2)}f(x,y_1,y_2)=&D_uf_{23}(x+v)+(y_1+y_2+1) D_u f_{1234}(x+v)\\
			=&(y_1+y_2)(D_uf_{23}(x+v))+(y_1+y_2+1)D_u f_{14}(x+v)\\
			=&(y_1+y_2)(D_uf_2(x)+D_uf_3(x+v))\\
			+&(y_1+y_2+1)(D_uf_1(x)+D_uf_4(x+v)).
		\end{split}
	\end{equation}
	
	\noindent Since for any $v\in  \F_2^n $ and any $V$, there exist $u^{(3)}\in V $ such that  $D_{u^{(3)}}f_2(x)+D_{u^{(1)}}f_3(x+v)\neq 0,~\textit{or}~D_{u^{(3)}}f_1(x)+D_{u^{(3)}}f_4(x+v)\neq 0$, from (\ref{equ the5.10 2 c3}), we have $$D_{(u^{(3)},0,0)}D_{(v,d_1,d_2)}f(x,y_1,y_2)\neq 0.$$
	In this way, we conclude that $f\notin\cM^\#$.
\end{proof}

In the special case when $f_4=f_1+f_2+f_3$,  we have the following corollary. 
\begin{cor}\label{cor:sharing_uniquesubspace}
Let $f_1, \ldots,f_4$ be four bent  functions on $\F_2^n$, with $n=2m$, satisfying the following conditions:
\begin{enumerate}
	\item $f_1, \ldots,f_4$ belong to $\cM^\#$ and  share a unique $\mathcal{M}$-subspace $U$;
	\item  $f=f_1||f_2||f_3||f_4 \in\B_{n+2}$ is a bent function.
\end{enumerate}
Let $V$ be an $(\frac{n}{2}-1)$-dimensional subspace  of $\F_2^n$  such that $D_aD_bf_i=0$, for all $a,b \in V$; $i=1, \ldots ,4.$
If for any $v\in  \F_2^n $ and  any such $V \subset \F_2^n$, there exist $u^{(1)},u^{(2)},u^{(3)}\in V $ such that the following three conditions hold simultaneously 
\begin{enumerate}
	\item[1.] $D_{u^{(1)}}f_1(x)+D_{u^{(1)}}f_2(x+v)\neq 0,$
	\item[2.] $D_{u^{(2)}}f_1(x)+D_{u^{(2)}}f_3(x+v)\neq 0,$
	\item[3.] $D_{u^{(3)}}f_2(x)+D_{u^{(3)}}f_3(x+v)\neq 0,$
\end{enumerate}
then $f$ is outside $\cM^\#$.
\end{cor}

\begin{cor}\label{cor:sharing_uniquesubspaceandn/2-1}
	With the same notation as in Theorem \ref{th:sharingcommonsubspace}, we assume  that $f_4=f_1+f_2+f_3$ and   $V\subset U$ for any $V$, where $\dim(V)=n-1$ and $U$ is a unique common $\mathcal{M}$-subspace of $f_1,f_2,f_3,f_4$. Then, the following set of sufficient conditions ensures that $f=f_1||f_2||f_3||f_4 \in \B_{n+2}$ does not belong to  $\cM^\#$:\\
	There exist one subspace $S\subset U$ with $\dim(S)=2$ such that 
	$$\begin{array}{c}D_{u}f_1(x)+D_{u}f_2(x+v)\neq 0;\\
	D_{u}f_1(x)+D_{u}f_3(x+v)\neq 0;\\
	
	D_{u}f_2(x)+D_{u}f_3(x+v)\neq 0,
\end{array}$$
for  any $u\in S\setminus \{0_n\}, v\in \F_2^n$.  
\end{cor}
\begin{proof}
If  we always have $V\subset U$ for any $V$, then  $\dim(V\cap S)\geq 1.$  This follows from the fact that  $\dim(S)=2, \dim(V)=n-1$ and furthermore $S\subset U$ and $ V\subset U$. 
Thus, for any $V$, we always can find at least one nonzero vector $u'\in V\cap S$. 
Since 	$$\begin{array}{c}D_{u}f_1(x)+D_{u}f_2(x+v)\neq 0;\\
	D_{u}f_1(x)+D_{u}f_3(x+v)\neq 0;\\
	
	D_{u}f_2(x)+D_{u}f_3(x+v)\neq 0,
\end{array}$$
for  any $u\in S\setminus \{0_n\}, v\in \F_2^n$, 
we have 
$$\begin{array}{c}D_{u'}f_1(x)+D_{u'}f_2(x+v)\neq 0;\\
	D_{u'}f_1(x)+D_{u'}f_3(x+v)\neq 0;\\
	
	D_{u'}f_2(x)+D_{u'}f_3(x+v)\neq 0.
\end{array}$$
From Theorem \ref{th:sharingcommonsubspace}, we know $ f$ is outside $\cM^\#$. 
\end{proof}

	\begin{ex}\label{ex: Outside everything 3}
	Consider the following Boolean bent functions $f_1,f_2,f_3,f_4\in\mathcal{B}_6$, which all belong to $\cM^\#$  and are given by algebraic normal form as follows:
	\begin{equation}\label{eq: Frobenius decomposition}
		\begin{split}
			f_1(x,y)= &   x_1 (y_2 + y_3 + y_1 y_3) + x_2 (y_1 + y_1 y_3 + y_2 y_3) + x_3 (y_1 y_2 + y_3) + y_1 + y_2 + y_3, \\
			f_2(x,y)=&  x_1 (y_2 + y_1 y_2 + y_1 y_3) + x_2 (y_1 + y_2 + y_1 y_2 + y_2 y_3) \\
			+& x_3 (y_1 + y_1 y_2 + y_3 + y_1 y_3 + y_2 y_3) + y_3 + 1, \\
			f_3(x,y)=&x_1 (y_1 + y_2 + y_1 y_2 + y_2 y_3) + x_2 (y_2 + y_3 + y_1 y_3) +  x_3 (y_1 + y_2 + y_3 + y_2 y_3)\\
			+&  y_2 + y_3 + 1,\\
			f_4(x,y)=&  x_1 (y_1 + y_2 + y_3 + y_2 y_3) + x_2 (y_1 y_2 + y_3) + x_3 (y_2 + y_3 + y_1 y_3) + y_1 + 1.
		\end{split}
	\end{equation}
	One can check that for defined in~\eqref{eq: Frobenius decomposition} bent functions, the dual bent condition is satisfied. In this way, we have that $f=f_1||f_2||f_3||f_4\in\mathcal{B}_8$ is bent. Its ANF is given by
	\begin{equation}
		\begin{split}
			f(z)&= z_4 + z_2 z_4 + z_5 + z_1 z_5 + z_3 z_4 z_5 + z_6 + z_1 z_6 + z_3 z_6 + z_1 z_4 z_6 + 
			z_2 z_4 z_6 + z_2 z_5 z_6 \\ & + z_7 + z_4 z_7 + z_1 z_4 z_7 + z_2 z_4 z_7 + z_3 z_4 z_7 + 
			z_2 z_5 z_7 + z_3 z_5 z_7 + z_1 z_4 z_5 z_7 + z_3 z_4 z_5 z_7 \\ & +  z_1 z_6 z_7  + 
			z_2 z_6 z_7 + z_1 z_4 z_6 z_7 + z_1 z_5 z_6 z_7 + z_2 z_5 z_6 z_7 + 
			z_3 z_5 z_6 z_7 + z_8 + z_4 z_8 \\ & + z_3 z_4 z_8  + z_5 z_8 + z_2 z_5 z_8 + 
			z_1 z_4 z_5 z_8 + z_2 z_4 z_5 z_8 + z_1 z_6 z_8 + z_2 z_4 z_6 z_8 + z_3 z_4 z_6 z_8 \\ & + 
			z_3 z_5 z_6 z_8 + z_7 z_8 + z_6 z_7 z_8.
		\end{split}
	\end{equation}
	Since every bent function $f_i$ has the form $f_i(x,y)=x\cdot\pi_i(y)+h_i(y)$, where $\pi_i$ is a quadratic APN permutation, then $f_i$ share the unique canonical $\mathcal{M}$-subspace $U=\F_2^3\times\{0_3\}$. In this way, we cannot use Theorem~\ref{th: two MM with a unique M-subspace}. One can check that for every two-dimensional subspace  $V$ of $\F_2^8$ such that $D_aD_bf_i=0$, for all $a,b \in V$, where $i=1, \ldots ,4$, the conditions of Theorem~\ref{th:sharingcommonsubspace} are satisfied, and hence, the bent function $f=f_1|| f_2|| f_3||f_4\in\mathcal{B}_{8}$ is outside $\cM^\#$. Additionally, using Algorithm~\ref{algorithm: Membership in the partial spread class}, we confirm that $f\notin\mathcal{PS}^\#$, and, hence, $f\notin(\cM^\#\cup\mathcal{PS}^\#)$.
\end{ex}
\begin{rem}
	The examples in this section indicate that concatenation $f=f_1||f_2||f_3||f_4$ of four bent functions $f_i\in\cM^\#$ can give a new bent function $f\notin(\cM^\#\cup\mathcal{PS}^\#)$. We would also like to note that all functions $f\in\mathcal{B}_8$ obtained in Examples~\ref{ex: Outside everything 1}, \ref{ex: Outside everything 2} and \ref{ex: Outside everything 3} are pairwise inequivalent. The latter was checked with Magma using the design isomorphism, as described in~\cite{Polujan2020}.
\end{rem}
The examples in this section indicate, that proper concatenations of bent functions satisfying the dual bent condition can give rise to many instances of (inequivalent) bent functions outside $\cM^\#$. This observation motivates the following research problem.
\begin{op}
	Find bent functions $f_1,f_2,f_3,f_4\in\mathcal{B}_{n}$ satisfying the dual bent condition, i.e., $f_1^*+ f_2^*+ f_3^*+f_4^*=1$, such that $f=f_1|| f_2|| f_3||f_4\in\mathcal{B}_{n+2}$ is bent and outside $\cM^\#$.
\end{op}

\section{Conclusion and open problems}\label{sec: concl}
In this article we have analyzed the structure of bent functions in the Maiorana-McFarland class with respect to their inherent $\mathcal{M}$-subspaces, thus contributing to the analysis of inequivalent Maiorana-McFarland bent functions. Moreover, we provided generic construction methods of bent functions outside $\cM^\#$ for any $n \geq 8$ using the bent 4-concatenation. Most notably, our results indicate that it is possible to construct bent functions outside $\cM^\#\cup\mathcal{PS}^\#$, thus we contribute to the better understanding of the origin of bent functions in $n=8$ variables.

To conclude, we believe that answering the following questions (in addition to the already mentioned open problems) will help to shed more light on the classification of bent functions as well as to develop new generic construction methods of these functions.
\begin{itemize}
	\item[1)] As we mentioned in the introduction, for a Maiorana-McFarland bent function $f\in\mathcal{B}_n$, the number of its $\mathcal{M}$-subspaces is at most $\prod_{i=1}^{n/2} \left(2^i+1\right)$ and the equality is attained if and only if $f$ is quadratic. What is the maximum number of $\mathcal{M}$-subspaces for a bent function $f\in\mathcal{B}_n$ in $\cM$ of a fixed degree $d>2$, and is it possible to characterize the functions achieving this bound? Our computational results indicate, that bent functions of the form $(x,y)\mapsto x\cdot y + y_{i_1}y_{i_1}\cdots y_{i_d}$ have the maximum number of $\mathcal{M}$-subspaces among all Maiorana-McFarland bent function of a fixed degree $d>2$.
	\item[2)] In this article, we analyzed, which properties of permutations $\pi$ guarantee that Maiorana-McFarland bent functions $x\cdot \pi(y)+h(y)$ have either one or many  $\mathcal{M}$-subspaces. For example, if $\pi$ has the $\eqref{eq: P1}$ property, we know that independently of the choice of the function $h$, the bent function $x\cdot \pi(y)+h(y)$ has the unique canonical $\mathcal{M}$-subspace. However, if the $\eqref{eq: P1}$ property is relaxed, then the properties of the function $h$ become crucial to guarantee the uniqueness of the $\mathcal{M}$-subspace. We think it is important to understand in general, how the choice of a pair $(\pi,h)$ affects the number of $\mathcal{M}$-subspaces of the corresponding Maiorana-McFarland function.
	\item[3)] An efficient way to satisfy the dual bent condition (we have to ensure that $f_1^* + f_2^* +f_3^* +f_4^*=1$ so that $f=f_1||f_2||f_3||f_4$ is bent) is to use $f_1=f_2$ and $f_3=1+f_4$  which we employed in Theorem \ref{th: two MM with a unique M-subspace}. However, there exist other possibilities to satisfy the dual bent condition which need to be examined further with  regard to the class membership of the designed bent functions. 
	We notice that Proposition \ref{prop:commonsubspace} does not require that the functions $f_i$ that define $f=f_1||f_2||f_3||f_4$ are bent.
	Therefore, another interesting research problem is to apply a similar approach  as taken  in Theorem  \ref{th: two MM with a unique M-subspace} to semi-bent and 5-valued spectra functions.  
\end{itemize}
\section*{Acknowledgements} 
Enes Pasalic is supported in part by the Slovenian Research Agency (research program P1-0404 and research projects J1-1694, N1-0159, J1-2451 and J1-4084). Sadmir Kudin is supported in part by the Slovenian Research Agency (research program P1-0404, research project J1-4084 and Young Researchers Grant). Fengrong Zhang is supported in part by the Natural Science Foundation of China (No. 61972400), the Fundamental Research Funds for the Central Universities (XJS221503), and the Youth Innovation Team of Shaanxi Universities.





\end{document}